\documentclass[11pt]{article}
\usepackage{amsmath,amssymb,amsthm}
\usepackage[noend]{algpseudocode}
\usepackage{algorithm}
\usepackage{graphicx}
\usepackage{fullpage}
\usepackage{url}
\usepackage{enumitem}
\usepackage[top=1in, bottom=1in, left=1in, right=1in]{geometry}
\usepackage{ifthen}
\usepackage{subcaption}
\usepackage{mathtools}
\usepackage{cleveref}

\newcommand{\dist}{\mathrm{\mathbf{d}}}
\newcommand{\M}{\mathcal{M}}
\newcommand{\ball}{\mathbf{B}}

\newcommand{\parent}{\mathrm{par}}
\newcommand{\child}{\mathrm{ch}}
\newcommand{\rel}{\mathrm{rel}}

\newcommand{\nt}{\mathrm{NT}}
\newcommand{\wnt}{\mathrm{LNT}}
\newcommand{\h}{\mathrm{h}}

\newcommand{\vor}{\mathrm{Vor}}
\newcommand{\spread}{\Delta}

\newcommand{\centersite}{\mathcal{C}}
\newcommand{\cell}{\mathcal{S}}
\newcommand{\bunch}{\mathcal{B}}

\newcommand{\splitaboveevent}{\Phi}
\newcommand{\splitbelowevent}{\Psi}
\renewcommand{\because}[1]{& \left[\text{#1}\right]}

\newtheorem{theorem}{Theorem}
\newtheorem{lemma}[theorem]{Lemma}

\newtheorem{definition}[theorem]{Definition}
\newtheorem*{invariant*}{Invariant}

\title{Randomized Incremental Construction of Net-Trees}
\author{
  Mahmoodreza Jahanseir\thanks{University of Connecticut \texttt{reza@uconn.edu}}
  \and
  Donald R. Sheehy\thanks{University of Connecticut \texttt{don.r.sheehy@gmail.com}}
}
\date{}

\begin{document}  \maketitle
  \setcounter{page}{0}
  \begin{abstract}
  Net-trees are a general purpose data structure for metric data that have been used to solve a wide range of algorithmic problems.
  We give a simple randomized algorithm to construct net-trees on doubling metrics using $O(n\log n)$ time in expectation.
  Along the way, we define a new, linear-size net-tree variant that simplifies the analyses and algorithms.
  We show a connection between these trees and approximate Voronoi diagrams and use this to simplify the point location necessary in net-tree construction.
  Our analysis uses a novel backwards analysis that may be of independent interest.
\end{abstract}

  \thispagestyle{empty}
  \clearpage
  \section{Introduction}
\label{sec:Introduction}

  Har-Peled \& Mendel introduced the net-tree as a linear-size data structure that efficiently solves a variety of (geo)metric problems such as approximate nearest neighbor search, well-separated pair decomposition, spanner construction, and others~\cite{har-peled06fast}.
  More recently, such data structures have been used in efficient constructions for topological data analysis (TDA)~\cite{sheehy13linear}.
  Net-trees are similar to several other data structures that store points in hierarchies of metric nets (subsets satisfying some packing and covering constraints) arranged into a tree or DAG.
  Examples include navigating nets~\cite{krauthgamer04navigating}, cover trees~\cite{beygelzimer06cover}, dynamic hierarchical spanners~\cite{cole06searching,gottlieb08optimal}, and deformable spanners~\cite{gao06deformable}.

  The extensive literature on such data structures can be partitioned into two disjoint groups: those that are easy to implement and those that can be constructed in $O(n\log n)$ time for doubling metrics (see Section~\ref{sec:preliminaries} for the definition).
  In this paper, we present an algorithm that is both simple and asymptotically efficient.
  We combine several ideas already present in the literature with a randomized incremental approach.
  The challenge is relegated to the analysis, where the usual tricks for randomized incremental algorithms do not apply to net-trees, mostly because they are not canonically defined by a point set.

  There are two known algorithms for building a net-tree~\cite{har-peled06fast} or a closely related structure~\cite{gottlieb08optimal} in $O(n\log n)$ time for doubling metrics.
  Both are quite complex and are primarily of theoretical interest.
  The algorithm of Har-Peled \& Mendel~\cite{har-peled06fast} requires a complex sequence of approximating data structures.
  Cole \& Gottlieb~\cite{cole06searching} proposed a similar data structure that supports dynamic insertions and deletions in $O(\log n)$ time.
  Their data structure maintains a so-called centroid path decomposition of a spanning tree of the hierarchy into a collection of paths that are each represented by a biased skip list.

  A much simpler algorithm due to Clarkson~\cite{clarkson02nearest} can be combined with an algorithm of Har-Peled \& Mendel~\cite{har-peled06fast} to run in $O(n\log \spread)$ time, where $\spread$ is the the \emph{spread} of the input, i.e.\ the ratio of the largest to smallest pairwise distances.
  Most of the complications of the theoretical algorithm are to eliminate this dependence on the spread.
  The goal of this paper is to combine the conceptual simplicity of Clarkson's idea with a simple randomized incremental algorithm to achieve the same $O(n\log n)$ running time of the best theoretical algorithms.

  The main improvement over the related data structures~\cite{cole06searching,har-peled06fast} that can be computed in $O(n\log n)$ time is the increased simplicity.
  In particular, the point location structure, the primary bottleneck in all related work is just a dictionary mapping uninserted points to nodes in the tree and is very easy to update.
  As testament to the simplicity, a readable implementation in python is only $\sim200$ lines of code and is available online~\cite{githubcode}.

  A second improvement is in the tighter bounds on the so-called relative constant.
  This is the constant factor that bounds the ratio of distances between relative links, the edges stored between nearby nodes in the same level of the tree.
  These relative links form a hierarchical spanner, and their number dominates the space complexity of the data structure.
  In similar constructions used in TDA, it was found that although a relative constant of $10$ was needed for a particular algorithm, the space blowup required using a constant closer to $3$ in practice, sacrificing theoretical guarantees~\cite{oudot14zigzag}.
  In previous work the relative constant was $13$ or more.
  In this work, we show that the relative constant can be pushed towards $2$ as a function of the difference in scales between adjacent levels in the tree.
  A value of $6$ is easily achievable in practice.

  \paragraph*{Related Work}
  Uhlmann~\cite{uhlmann91satisfying} proposed \emph{metric trees} to solve range searches in general metric spaces, but there are no performance guarantees on the queries.
  Yianilos~\cite{yianilos93data} devised a similar data structure, called the \emph{vp-tree}, and he showed that when the search radii are very small, queries can be run in $O(\log n)$ expected time.
  These structures are balanced binary search trees and they can be constructed recursively in $O(n\log n)$ time by partitioning the points into two subsets according to their distance to the median.

  Clarkson~\cite{clarkson99nearest} proposed two randomized data structures to answer approximate nearest neighbor queries in metric spaces satisfying a sphere packing property, which is equivalent to having constant doubling dimension.
  The first data structure assumes that the query points have the same distribution as the given input points and it may fail to return a correct answer.
  The second one always returns a correct answer, but it requires more time and space.
  Roughly speaking, both of these data structures can be constructed in $O(n\log\spread)$ time, with $O(n\log\spread)$ size, and they answer queries in $O(\log\spread)$ time.

  Karger \& Ruhl~\cite{karger02finding} proposed \emph{metric skip lists} for the so-called growth restricted metrics.
  Their data structure can be constructed in $O(n\log n\log\log n)$ time and has $O(n\log n)$ size.

  Krauthgamer \& Lee~\cite{krauthgamer04navigating} presented a dynamic, deterministic data structure called \emph{navigating nets} to address proximity searches in doubling spaces.
  Navigating nets are comprised of hierarchies of nested metric nets connected as a DAG.
  In more detail, the points at some scale $i$ are of distance $2^i$ and the balls centered at those points with radii $2^i$ contain all points of scale $i-1$.
  They showed that navigating nets have linear size and can be constructed in $O(n\log\spread)$ time.
  Gao et al.~\cite{gao06deformable} independently devised a very similar data structure called the \emph{deformable spanner} as a dynamic $(1+\epsilon)$-spanner in Euclidean spaces that can be maintained under continuous motion.

  Beygelzimer et al.~\cite{beygelzimer06cover} proposed the \emph{cover tree}, a spanning tree of a navigating net, to make the space independent of the doubling dimension.
  Their experimental results showed that cover trees have good performance in practice.
  Besides the space complexity, cover trees do not theoretically outperform navigating nets.
  Recently, we~\cite{jahanseir16transform} showed that cover trees with slight modifications also satisfy the stricter properties of net-trees, and a net-tree can be constructed from a cover tree in linear time assuming the space may depend on the doubling dimension.
  Refer to~\cite{clarkson06nearest,chavez01searching} for surveys on proximity searches in metric spaces.

  \paragraph*{Overview of the Algorithm and its Analysis}
  Our approach is to construct a net-tree incrementally.
  Each new point is added by first attaching it as a leaf of the tree.
  Then, new nodes are added for that point by propagating up the tree one level at a time using just local updates until the covering property is satisfied.
  We call this the \emph{bottom-up insertion}.
  By making only local updates, we can show that the total work of updating the tree is linear in the size and thus, linear in the number of points.
  This algorithm relies on a point location data structure that associates each uninserted point with a node in the tree called its \emph{center}, which provides a starting point for the bottom-up insertion.
  Each time a new node is added, a local search is required to see if any of the uninserted points should have the new node as their center.
  The challenge is to show that the expected total number of distance computations for the point location data structure is only $O(n\log n)$ for points inserted in random order.

  Our approach to point location is similar in spirit to that proposed by Clarkson~\cite{clarkson02nearest}.
  The difference is that instead of associating uninserted points with Voronoi cells of inserted points, we associate them with nodes in the tree.
  This allows us to work with arbitrary (and random) orderings of the points rather than just greedy permutations.

  Our analysis includes a novel backwards analysis on random orderings to bound the expected running time.
  The main difficulty is that the tree structure is far from canonical for a given set of points.
  Instead of looking at the tree construction backwards, we define a set of random events that can occur in a permutation of metric points and then charge the work of point location to these events.
  We show that each event is charged only a constant number of times and there are at most $O(n\log n)$ events in expectation.
  The resulting expected running time of $O(n\log n)$ matches the best theoretical algorithms.

  \section{Preliminaries}
\label{sec:preliminaries}

  \subsection{Doubling Metrics and Packing}
  \label{sub:doubling_metrics}

  The input is a set of $n$ points $P$ in a metric space.
  The \emph{closed metric ball} centered at $p$ with radius $r$ is denoted $\ball(p,r):=\{q\in P \mid \dist(p,q)\le r\}$.
	The \emph{doubling constant} $\rho$ of $P$ is the minimum $\rho$ such that every ball $\ball(p,r)$ can be covered by $\rho$ balls of radius $r/2$.
  We assume $\rho$ is constant.
  The \emph{doubling dimension} is defined as $\lg\rho$.
  A metric space with a constant doubling dimension is called a \emph{doubling metric}.
  Throughout, we assume that the input metric is doubling.
  The doubling dimension was first introduced by Assouad~\cite{assouad83plongements} and has since found many uses
  in algorithm design and analysis.
  Other notions of dimension for general metric spaces have also been proposed.
  The notion of growth-restricted metrics of Karger \& Ruhl~\cite{karger02finding} is similar to doubling
  metrics, though it is more restrictive.
  Gupta et al.~\cite{gupta03bounded} showed that the dimension of a growth restricted metric is upper bounded by its doubling dimension.

  Perhaps the most useful property of doubling metrics is that they allow for the use of packing and covering arguments similar to those used in Euclidean space to carry over to a more general class of metrics.
  The following lemma is at the heart of all the packing arguments in this paper.

  \begin{lemma}[Packing Lemma]\label{lem:packing}
    If $X\subseteq\ball(p,r)$ and for every two distinct points $x,y\in X$, $\dist(x,y)>r'$, where $r>r'$, then $|X|\le\rho^{\lfloor\lg (r/r')\rfloor+1}$.
  \end{lemma}
  \begin{proof}
    By the definition of the doubling constant, $X$ can be covered by $\rho$ balls of radius $r/2$.
    These balls can each be covered by $\rho$ balls of radius $r/4$.
    Repeating this $\lfloor\lg (r/r')\rfloor+1$ times results in at most $\rho^{\lfloor\lg (r/r')\rfloor+1}$ balls of radius less than $r'$, and these balls contain at most one point of $X$ each, so $X$ has at most $\rho^{\lfloor\lg (r/r')\rfloor+1}$ points.
  \end{proof}

  The spread of a point set $P$ often plays a role in running times of metric data structures.
  The following lemma captures the relationship between the spread and the cardinality of $P$ and follows directly from the Packing Lemma.
  \begin{lemma}
    A finite metric $P$ has at most $\rho^{O(\log\Delta)}$ points.
  \end{lemma}

  The distance from a point $p$ to a compact set $Q$ is defined as $\dist(p, Q) \coloneqq \min_{q\in Q} \dist(p,q)$.
  The \emph{Hausdorff distance} between two sets $P$ and $Q$ is $\dist_H(P,Q)\coloneqq\max\{\max_{p\in P} \dist(p,Q),\max_{q\in Q} \dist(q,P)\}$.

  \subsection{Metric Nets and Net-Trees}
  \label{sub:net_trees}
  A metric net or a \emph{Delone set} is a subset of points satisfying some packing and covering properties.
  More formally, an $(\alpha,\beta)$-net is a subset $Q\subseteq P$ such that: for all distinct points $p,q\in Q$, $\dist(p,q)>\alpha$ (packing), and for all $p\in P$, $\dist(p, Q)\le \beta$ (covering).

  \begin{figure}[htb]
    \centering
    \includegraphics[width=\columnwidth]{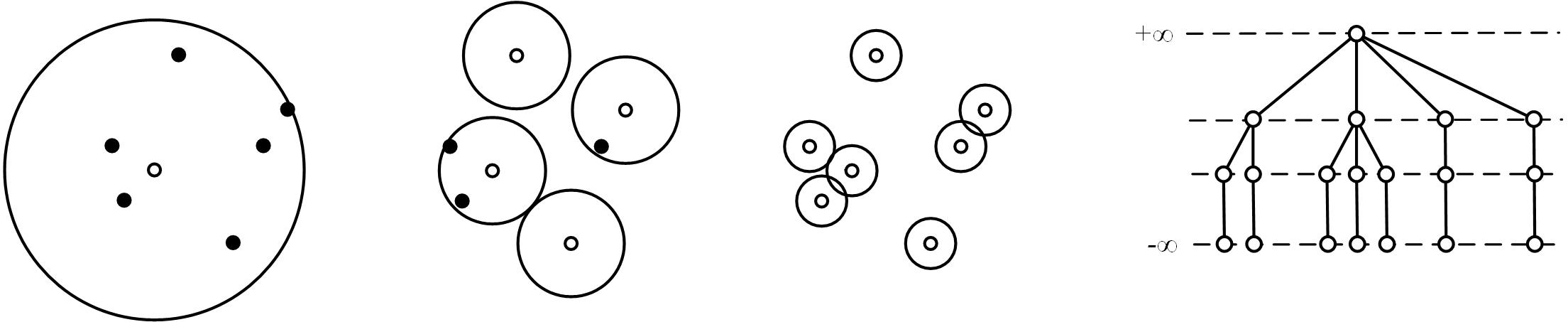}
    \caption{Nets at three different scales are shown from the left and the corresponding net-tree is illustrated on the right.
    White dots represent a net and the circles show the covering balls.}
    \label{fig:nets_and_nettree}
  \end{figure}

  A net-tree is a tree $T$ in which each level represents a metric net at some scale, see Fig.~\ref{fig:nets_and_nettree}.
  In net-trees, points are leaves in level $-\infty$ and each point can be associated with many internal nodes.
  Each node is uniquely identified by its associated point and an integer called its \emph{level}.
  The node in level $\ell$ associated with a point $p$ is denoted $p^\ell$.
  We assume that the root is in level $+\infty$.
  For a node $p^\ell\in T$, we define $\parent(p^\ell)$ and $\child(p^\ell)$ to be the parent and the set of children of that node, respectively.
  Let $P_{p^\ell}$ denote leaves of the subtree rooted at $p^\ell$.
  For each node $p^\ell$ in a net-tree, the following properties hold.
  \begin{itemize}
    \item \textbf{Packing:} $\ball(p, c_p \tau^\ell)\bigcap P \subseteq P_{p^\ell}$.
    \item \textbf{Covering:} $P_{p^\ell}\subset \ball(p, c_c\tau^{\ell})$.
    \item \textbf{Nesting:} If $\ell>-\infty$, then $p^\ell$ has a child with the same associated point $p$.
  \end{itemize}

  The constant $\tau>1$, called the \emph{scale factor}, determines the change in scale between levels.
  We call $c_p$ and $c_c$ the \emph{packing constant} and the \emph{covering constant}, respectively, and $c_c\ge c_p>0$.
  We represent all net-trees with the same scale factor, packing constant, and covering constant with $\nt(\tau,c_p,c_c)$.
  From the above definition, Har-Peled \& Mendel showed that each level of a net-tree is a metric net~\cite{har-peled06fast}.

  There are two different representations for net-trees.
  In the \emph{uncompressed} representation, every root to leaf path has a node in every level down to the scale of the smallest pairwise distance.
  The size complexity of this representation is $O(n\log\Delta)$, because there are $O(\log\Delta)$ explicit levels between $-\infty$ and $+\infty$.
  The $compressed$ representation is obtained from the uncompressed one by removing the nodes that are the only child of their parents and they have only one child and merging the two adjacent edges as a long edge, see Fig.~\ref{fig:SCNT}.
  We call such long edges \emph{jumps}.
  It is not hard to see that this representation has size of $O(n)$.
  Note that compressed net-trees are similar to compressed quadtrees.

  A net-tree can be augmented to maintain a list of nearby nodes called relatives.
  We define relatives of a node $p^\ell\in T$ to be
  \begin{align*}
    \rel(p^\ell)\coloneqq\{x^f\in T \text{ with } y^g\coloneqq\parent(x^f) \mid f \le \ell < g, \text{ and } \dist(p,x)\le c_r\tau^\ell\},
  \end{align*}
  see Fig.~\ref{fig:SCNT}.
  We call $c_r$ the \emph{relative constant}, and it is a function of the other parameters of a net-tree.
  In this paper, we assume that net-trees are always equipped with relatives.

  Har-Peled \& Mendel defined compressed net-trees in the class of $\nt(\tau=11,\frac{\tau-5}{2(\tau-1)},\frac{2\tau}{\tau-1})$ with $c_r=13$.
  The following easy to prove lemma uses the Packing Lemma and the definition of net-trees. It implies that a compressed net-tree on a doubling metric has $\rho^{O(1)}n$ size.

  \begin{lemma}
  \label{lem:ntsize}
    For each node $p^\ell$ in $T\in\nt(\tau,c_p,c_c)$, we have $|\child(p^\ell)|\le\rho^{\lfloor\lg (c_c\tau/c_p)\rfloor+1}$ and $|\rel(p^\ell)|\le\rho^{\lfloor\lg (c_r/c_p)\rfloor+1}$.
  \end{lemma}

  We defined $\child(\cdot)$, $\parent(\cdot)$, and $\rel(\cdot)$ for a node of a tree; however, we abuse notation slightly and apply them to set of nodes.
  In such cases, the result will be the union of output for each node.
  Furthermore, the distance between nodes of a net-tree is the distance between their corresponding points.

  \section{Net-Tree Variants} 
\label{sec:variations_of_nettrees}

  In this section, we introduce two natural modifications to net-trees that simplify both construction and analysis.
  In the first variant, we replace the global packing and covering conditions of a net-tree with local ones that are easier to check, and we show that these local conditions imply the global conditions.
  In the second variant, we show how a less aggressive compression criterion still results in a linear-size data structure while guaranteeing that relatives are on the same level in the tree, are symmetric, and are consistent up the tree (i.e.\ parents of relatives are relatives).
  This makes it much simpler to reason about local neighborhoods by local search among relatives.

  \subsection{Local Net-Trees} 
\label{sec:local_nettrees}

  Here, we define a local version of net-trees and we show that for some appropriate parameters, a local net-tree is a net-tree.
  The ``nets'' in a net-tree are the subsets
  \begin{align*}
    N_\ell\coloneqq\{p\in P\mid p^m\in T \text{ for some } m\ge \ell\}.
  \end{align*}

  A \emph{local net-tree} $T\in\wnt(\tau,c_p,c_c)$ satisfies the nesting property and the following invariants.
  \begin{itemize}
    \item \textbf{Local Packing:} For distinct $p,q\in N_\ell$, $\dist(p,q)>c_p\tau^\ell$.
    \item \textbf{Local Covering:} If $p^\ell=\parent(q^m)$, then $\dist(p,q)\le c_c\tau^{m+1}$.
    \item \textbf{Local Parent:} If $p^\ell=\parent(q^m)$, then $\dist(p,q)=\dist(p,N_{m+1})$.
  \end{itemize}
  The difference between the local net-tree invariants and the net-tree invariants given previously, is that there is no requirement that the packing or covering respect the tree structure.
  It is easy to see that the local packing and local covering properties can be obtained from the stronger ones.
  We are interested in local packing and covering properties because they are much easier to maintain as invariants after each update operation on a tree and also to verify in the analysis.

  The switch to local net-trees comes at the cost of having slightly different constants.
  Theorem~\ref{thm:local_to_strong_nettree} gives the precise relationship.

  \begin{lemma}
  \label{lem:covering}
    For all $x\in P_{p^\ell}$ in $T\in\wnt(\tau,c_p,c_c)$, $\dist(p,x)<\frac{c_c\tau}{\tau-1}\tau^\ell$.
  \end{lemma}
  \begin{proof}
  	By the local covering property and the triangle inequality, $\dist(p,x)\le \sum_{i=0}^{+\infty}c_c\tau^{\ell-i}\le \frac{c_c\tau}{\tau-1}\tau^\ell$.
  \end{proof}

  \begin{theorem}
  \label{thm:local_to_strong_nettree}
    For $\tau>\frac{2c_c}{c_p}+1$ and $0<c_p\le c_c<\frac{c_p(\tau-1)}{2}$, if $T\in\wnt(\tau,c_p,c_c)$, then $T\in\nt(\tau, \frac{c_p(\tau-1)-2c_c}{2(\tau-1)}, \frac{c_c\tau}{\tau-1})$.
  \end{theorem}
  \begin{proof}
    The covering property can be proved using Lemma~\ref{lem:covering}.
    To prove the packing property, let $p^\ell$ be a node of the local net-tree and $x\notin P_{p^\ell}$.
    Also, let $z^f$ be the lowest ancestor of $x$ with $f\ge\ell$, and $y^g$ be the highest ancestor of $x$ with $g<\ell$.
    It is clear that $y^g\in\child(z^f)$.
    Note that if $f>\ell$, then the edge between $x^f$ and $y^g$ is a jump and as a result $x$ and $y$ are the same points.
    From the local parent property, $\dist(y,z)<\dist(y,p)$.
    By the triangle inequality, $\dist(p,z)\le\dist(p,y)+\dist(y,z)<2\dist(p,y)$.
    Also, by the local packing property, $\dist(p,z)>c_p\tau^\ell$.
    The last two inequalities imply $\dist(p,y)>c_p\tau^\ell/2$.
    Since $g\le\ell-1$, by Lemma~\ref{lem:covering}, $\dist(x,y)\le c_c\tau^\ell/(\tau-1)$.
    By the triangle inequality,
    \begin{align*}
      \dist(p,x)\ge\dist(p,y)-\dist(y,x)>\frac{1}{2}c_p\tau^\ell-\frac{c_c}{\tau-1}\tau^\ell>\frac{c_p(\tau-1)-2c_c}{2(\tau-1)}\tau^\ell.
    \end{align*}
    Therefore, $P_{p^\ell}\subset\ball(p,\frac{c_p(\tau-1)-2c_c}{2(\tau-1)}\tau^\ell)$.
  \end{proof}

  If $c_c=c_p=1$, then a local net-tree with $\tau>3$ belongs to $\nt(\tau,\frac{\tau-3}{2(\tau-1)}, \frac{\tau}{\tau-1})$, which results in a definition of net-trees similar to Har-Peled \& Mendel's~\cite{har-peled06fast}.


  \subsection{Semi-Compressed Net-Trees} 
\label{sec:semi_compressed}

  In this section, we define \emph{semi-compressed} net-trees
  and show that this intermediate structure between uncompressed and compressed net-trees has linear size.
  The resulting graph of relatives is easier to work with because edges are undirected and stay on the same level of the tree.

  Recall that for compressed net-trees, we remove a node (by compressing edges) if it is the only child of its parent and has only one child.
  In semi-compressed net-trees, we do not remove a node if it has any relatives other than itself.
  Fig.~\ref{fig:SCNT} illustrates different representations of a net-tree on a set of points on a line.
  In the following theorem, we show that the semi-compressed representation has linear size.

  \begin{figure}[!tbh]
    \centering
    \includegraphics[width=0.31\columnwidth]{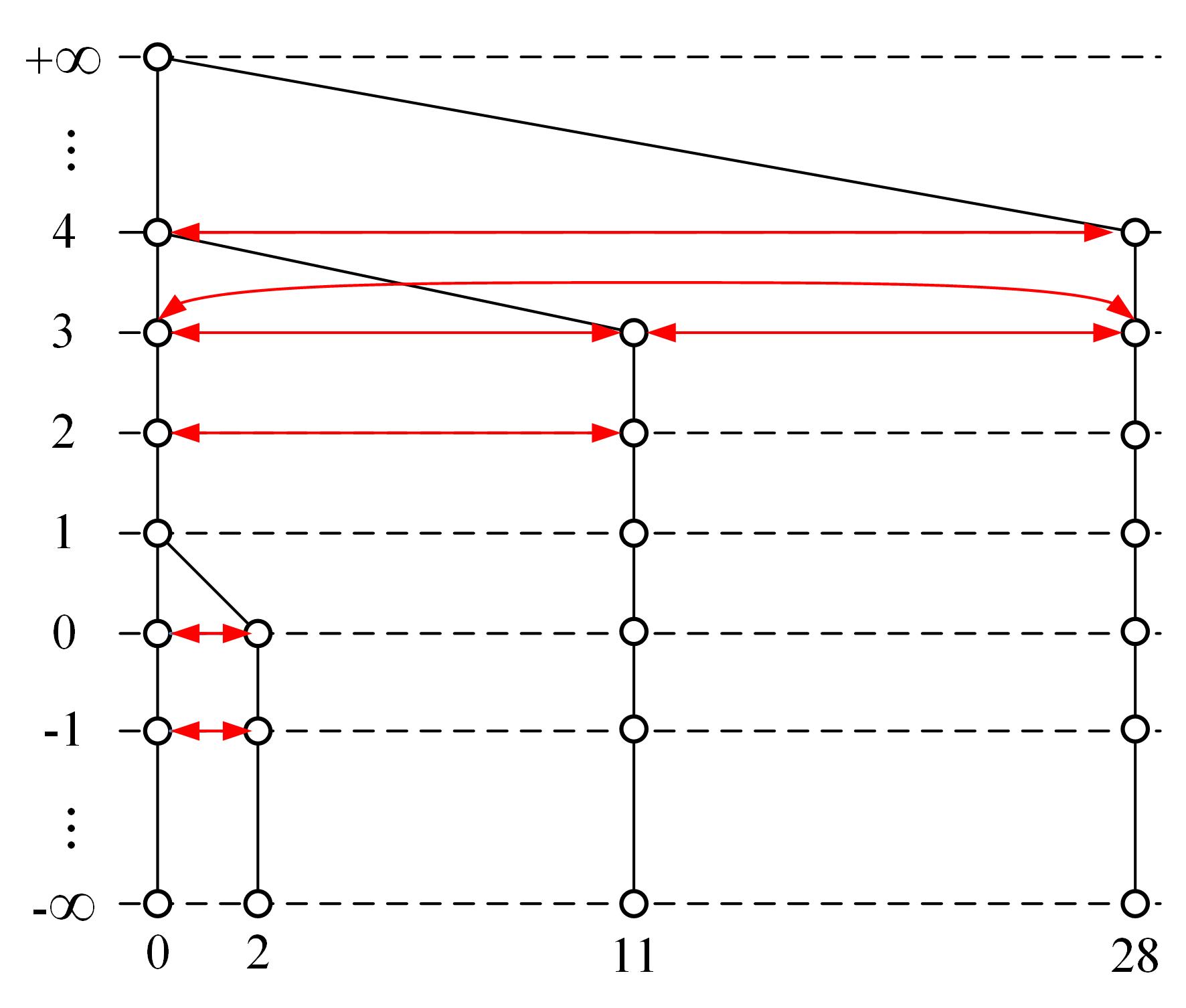}
    \label{fig:uncompressedNT}
    ~
    \includegraphics[width=0.31\columnwidth]{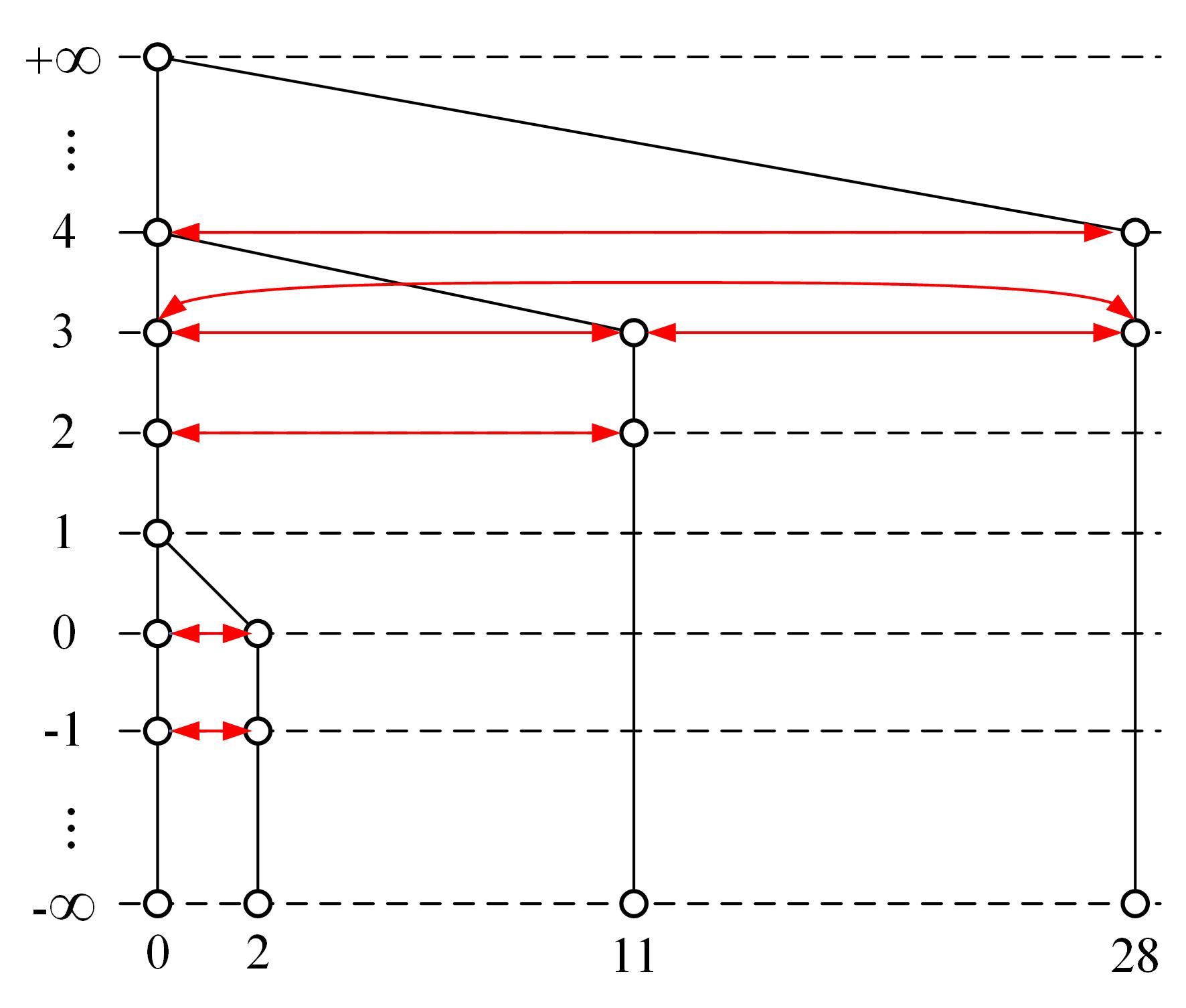}
    \label{fig:semicompressedNT}
    ~
    \includegraphics[width=0.31\columnwidth]{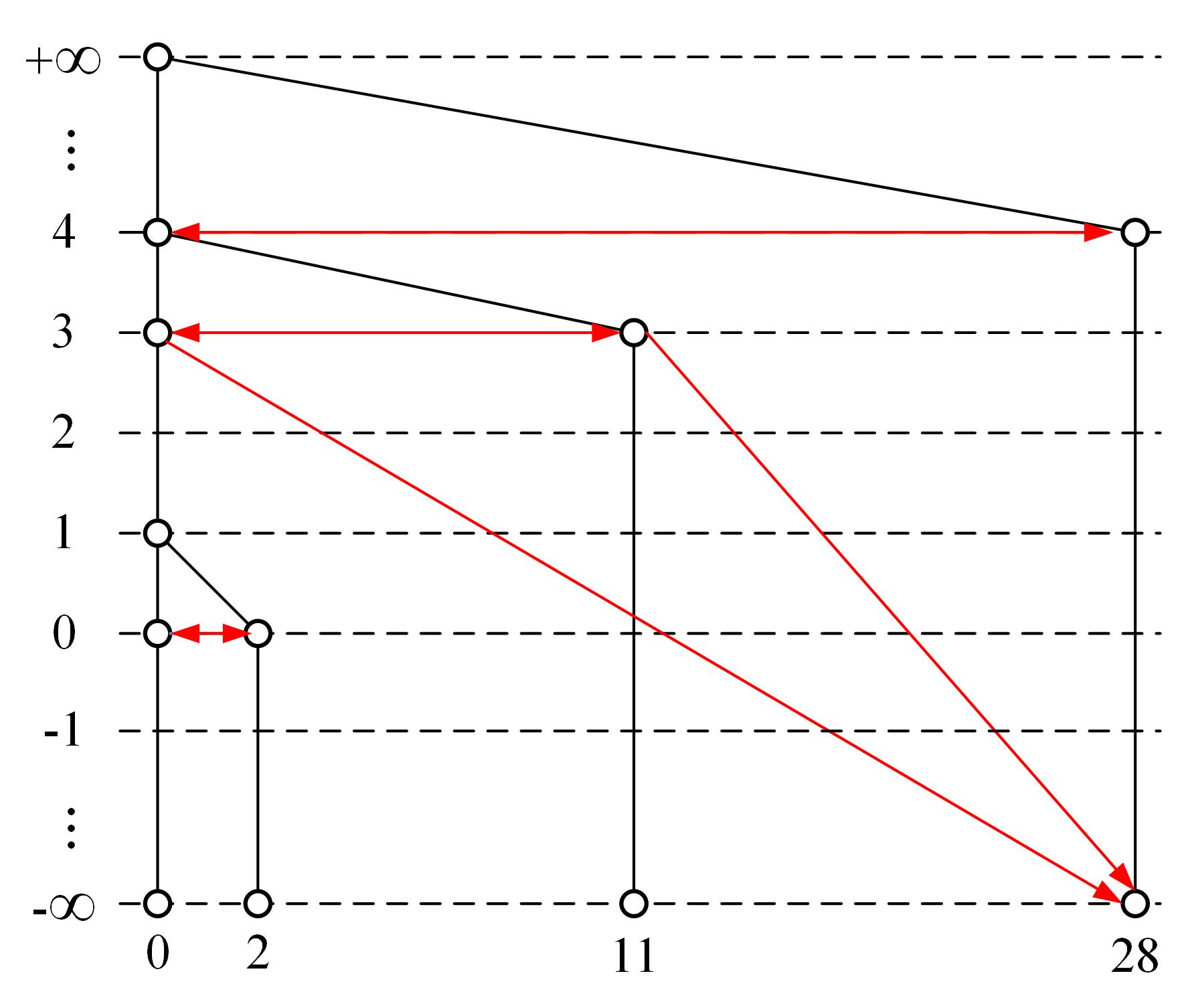}
    \label{fig:compressedNT}
    \caption{Different representations of a net-tree $T\in\nt(2,1,1)$ with $c_r=4$ on a set of points $\{0,2,11,28\}$ on a line.
    An arrow from node $a$ to node $b$ indicates $b$ is a relative of $a$.
    Left: uncompressed.
    Center: semi-compressed.
    Right: compressed.
    }
    \label{fig:SCNT}
  \end{figure}

  \begin{theorem}
    \label{thm:semi-compressed_linear_size}
    Given $n$ points $P$ in a doubling metric with doubling constant $\rho$. The size of a semi-compressed net-tree on $P$ is $O(\rho^{O(1)}n)$.
  \end{theorem}
  \begin{proof}
    Let $T$ be an uncompressed net-tree on $P$.
    Let $S$ be the semi-compressed tree formed from $T$.
    That is, $S$ contains the nodes of compressed net-tree and every $p^\ell\in T$ such that $|\rel(p^\ell)|>1$.
    It should be clear that the size of $S$ is $O(n + m)$ where $n$ is the number of input points and $m$ is the number of relative edges in the whole tree.
    So, if we can show $m = O(n)$, we will have shown that $S$ has linear size.

    First, we show that if two points $p$ and $q$ are relatives at some level $\ell$ in $T$, then they can be relatives in at most $\lceil\log_\tau (c_r/c_p)\rceil$ levels.
    Without loss of generality, let $\ell-i$ be the lowest level in $T$ such that $p$ and $q$ are relatives, where $i\ge0$.
    Then, $\dist(p,q)\le c_r\tau^{\ell-i}$.
    By the packing property at level $\ell$, $\dist(p,q)>c_p\tau^\ell$.
    Combining the last two inequalities results $c_p\tau^\ell<c_r\tau^{\ell-i}$.
    Therefore, $p$ and $q$ are relatives in $i\le\lceil\log_\tau (c_r/c_p)\rceil$ levels.

    Now, we find the number of relative edges.
    For a point $p$, we define $\h(p)\coloneqq \max\{\ell \mid p^\ell\in T\}$.
    If $p$ and $q$ are relatives in $T$, then we charge the point having $\min\{\h(p),\h(q)\}$ with $O(\log_\tau (c_r/c_p))$ to pay for the total number of relative edges between $p$ and $q$.
    Therefore, using Lemma~\ref{lem:ntsize}, $m = O(\sum_{p\in P} O(\log_\tau (c_r/c_p)) |\rel(p^{\h(p)})|)=O(\rho^{O(1)}n)$.
  \end{proof}

  In semi-compressed net-trees, the relative relation is symmetric, i.e. if $p^\ell\in\rel(q^\ell)$ then $q^\ell\in\rel(p^\ell)$, and we use $\sim$ to denote this relation.
  In the following lemma, we prove that if two nodes are relatives, then their parents are also relatives.
  \begin{lemma}
    \label{lem:functorial}
      In a semi-compressed net-tree $T\in\wnt(\tau,c_p,c_c)$ with $c_r\ge \frac{2c_c\tau}{\tau-1}$, if $p^\ell\sim q^\ell$, then $\parent(p^\ell)\sim \parent(q^\ell)$.
  \end{lemma}
  \begin{proof}
      Let $x^f\coloneqq\parent(p^\ell)$ and $y^g\coloneqq\parent(q^\ell)$.
      By the triangle inequality and the local covering property, $\dist(x,y)\le\dist(x, p)+\dist(p,q)+\dist(q,y)\le c_c\tau^{\ell+1}+c_r\tau^\ell+c_c\tau^{\ell+1}\le c_r\tau^{\ell+1}$ and $g=f=\ell+1$.
  \end{proof}



  \section{Approximate Voronoi Diagrams from Net-Trees}
\label{sec:approx_vd}

  Many metric data structures naturally induce a partition of the search space.
  The use of hierarchies of partitions at different scales is a fundamental idea in the \emph{approximate near neighbor} problem (also known as \emph{point location in equal balls} (PLEB))  which is at the heart of many \emph{approximate nearest neighbor} algorithms, including high dimensional approaches using locality-sensitive hashing~\cite{indyk98approximate,har-peled01replacement,indyk04nearest,sabharwal06nearest,har-peled12approximate}.

  Given a set of points $P$ and a query $q$, the \emph{nearest neighbor} of $q$ in $P$ is the point $p\in P$ such that for all $p'\in P$, we have $\dist(q,p)\le \dist(q,p')$.
  Relaxing this notion, $p$ is a \emph{$c$-approximate nearest neighbor} (or $c$-ANN) of $q$ if for all $p'\in P$, we have $\dist(q,p)\le c\dist(q,p')$.

  The Voronoi diagram of a set of points $P$ is a decomposition of space into cells, one per point $p\in P$ containing all points for which $p$ is the nearest neighbor.
  The nearest neighbor search problem can be viewed as point location in a Voronoi diagram, though it is not necessary to represent the Voronoi diagram explicitly.

  In this section we give a particular decomposition of space, an approximate Voronoi diagram from a net-tree.
  The purpose is not to introduce a new approximate Voronoi diagram (there are several already~\cite{har-peled01replacement,sabharwal06nearest,arya02linear}), but rather to provide a clear description of the point location problem at the heart of our construction.
  Just as in Clarkson's \textbf{sb} data structure~\cite{clarkson02nearest}, we will keep track of what ``cell'' contains each uninserted point.
  However, instead of using the Voronoi cells, we will use the approximate cells described below.
  Moreover, instead of having one cell per point, we have one cell per node, thus we can simulate having a Voronoi diagram of a net at each scale.

  We want to associate points with the closest node in the tree that is close enough to be a relative.
  Ties are broken between nodes associated to the same point by always choosing the one that is lowest in the tree.
  Formally, we define the following function mapping a point of the metric space $\M$ and a node to a pair of numbers.
  \[
    f(x, p^\ell) :=
      \begin{cases}
        (\dist(x,p), \ell) & \text{if $\dist(x,p) \le c_r \tau^\ell$}\\
        (\infty, \infty) & \text{otherwise}
      \end{cases}
  \]
  The Voronoi cell of a node $p^\ell$ is then defined as
  \begin{align*}
    \vor(p^\ell) \coloneqq \{x\in \M\mid f(x,p^\ell) \le f(x, q^m) \text{ for all }  q^m\in T\},
  \end{align*}
  where ordering on pairs is lexicographical.
  For a point $q\notin P$, the \emph{center} for $q$ in $T$, denoted $\centersite(q)$, is the node $p^\ell\in T$ such that $q\in \vor(p^\ell)$.
  As we will see in Section~\ref{sec:the_bottomup_construction}, finding the center of a point is the basic point location operation required to insert it into the net-tree.
  Fig.~\ref{fig:approx_vd} illustrates the construction.
  \begin{figure}[htb]
    \centering
    \includegraphics[trim= 0 7in 0 0.5in,width=0.5\textwidth]{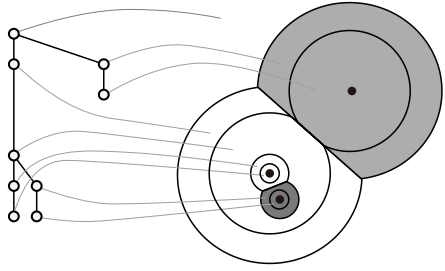}
    \caption{The net-tree on the left induces the approximate Voronoi diagram on the right.}
    \label{fig:approx_vd}
  \end{figure}

  The union of Voronoi cells $p^\ell$ for all $\ell$ gives an approximate Voronoi cell for the point $p$.
  The following lemma makes this precise.

  \begin{lemma}
  \label{lem:NN_center}
    Let $T$ be a net-tree in $\wnt(\tau,c_p,c_c)$ with $c_r>\frac{c_c\tau}{\tau-1}$ on a point set $P\subset \M$.
    For any point $q\in\M$, if $\centersite(q)=p^\ell$, then $p$ is a $(\frac{c_r\tau(\tau-1)}{c_r(\tau-1)-c_c\tau})$-ANN of $q$ in $P$.
  \end{lemma}
  \begin{proof}
    Let $m\coloneqq\lceil\log_\tau (\dist(p,q)/c_r)\rceil$.
    Then, $m\le\ell$ and $c_r\tau^{m-1}<\dist(p,q)\le c_r\tau^m$.
    Since $p\in N_m$ and $p^\ell$ is the center of $q$, $\dist(q,N_m)>c_r\tau^{m-1}$.
    Furthermore, $\dist(q,N_{m-1})>c_r\tau^{m-1}$, because otherwise $\centersite(q)$ should be a node other than $p^\ell$ so that the corresponding point belongs to $N_{m-1}$, which contradicts the assumption.
    Also note that each node associated to a point in $P\setminus N_{m-1}$ has an ancestor in a level at least $m-1$.
    If the lowest ancestor in a level at least $m-1$ is above $m-1$, then it is the top of a jump, and the bottom node with the same associated point is in a level less than $m-1$.
    Therefore, using Lemma~\ref{lem:covering}, $\dist_H(N_{m-1},P)\le c_c\tau^m/(\tau-1)$.
    Now, using the triangle inequality,
    \begin{align*}
      \dist(q,P)
      &\ge\dist(q,N_{m-1})-\dist_H(N_{m-1},P)
      > c_r\tau^{m-1}-\frac{c_c}{\tau-1}\tau^{m}
      >\left(\frac{1}{\tau}-\frac{c_c}{c_r(\tau-1)}\right)\dist(p,q).
    \end{align*}
    Therefore, $\dist(p,q)<\frac{c_r\tau(\tau-1)}{c_r(\tau-1)-c_c\tau}\dist(q,P)$.
  \end{proof}

  \section{Bottom-up Construction of a Net-Tree} 
\label{sec:the_bottomup_construction}

  Constructing a net-tree one point at a time has three phases.
  First, one finds the center (as defined in Section~\ref{sec:approx_vd}) of the new point.
  Second, the new point is inserted as a relative of its center, with its parent, children, and relatives computed by a constant-time local search.
  Third, new nodes associated with the point are added up the tree until the parent satisfies the covering property.
  In principle, this promotion phase can propagate all the way to the root.
  Along the way, it is sometimes necessary to split a compressed edge to create a node that now has a relative (our new point) or remove an existing node that now has no relatives.

  In the original work on net-trees, the difficult part of the algorithm finds not only the centers (or its equivalent), but also finds an ordering that avoids the propagation phase.
  Other algorithms have used the tree itself as the search structure to find the centers when needed~\cite{gao06deformable}, but this can lead to linear time insertions if the tree is deep.
  In this section, we will give the construction assuming the center of each new point is known, and we will describe the point location data structure in Section~\ref{sec:point_location}.

  \subsection{Insertion} 
  \label{sub:insertion}

    Once the center is found, $p$ is added to the tree as follows.
    Let $q^\ell\coloneqq\centersite(p)$.
    We find the lowest level $h$ in $T$ that $p$ has a relative (not itself).
    By the definition of relatives, $h\coloneqq\lceil\log_\tau (\dist(q,p)/c_r)\rceil$.
    If $p$ does not satisfy the packing property at level $h$, that is $\dist(p,q)\le c_p\tau^h$, then set $h\coloneqq h-1$.
    Next, we create node $p^h$.
    If $q^h$ is not already in the tree, then we add it to the tree.
    If the parent of $q^h$ is a node associated to point $q$ and $q^{h+1}\notin T$, then we create $q^{h+1}$ and add it to the tree.
    We also set the parent of $p^h$ to $\parent(q^h)$.

    To ensure that the parent, children, and relatives of the new node $p^h$ are correct, an update procedure will be executed.
    In this procedure, we find relatives and children of $p^h$ from $\child(\rel(\parent(p^h)))$ and $\child(\rel(p^h))$, respectively.
    Also, the parent of $p^h$ will be the closest node to $p$ among $\rel(\parent(p^h))$.
    Note that when node $p^h$ receives a new child, say $x^{h-1}$, we check the previous parent of $x^{h-1}$ against the semi-compressed condition to determine whether that node should be removed from the tree or not.
    The following lemma proves the correctness of the insertion algorithm.

    \begin{lemma}
    \label{lem:insert}
      Given a semi-compressed tree $T\in\wnt(\tau,c_p,c_c)$ with $c_r\ge\frac{2c_c\tau}{\tau-2}$ and an uninserted point $p$ with $q^\ell\coloneqq\centersite(p)$.
      The insertion algorithm adds $p$ into $T$ and results a semi-compressed local net-tree $T'\in\wnt(\tau,c_p,c_c+\frac{c_r}{\tau})$.
    \end{lemma}
    \begin{proof}
      We need to show that the resulted tree satisfies the covering, the packing, and the parent invariants, also relatives are correct and the output is semi-compressed.
      Let $x^{h+1}$ be the closest node to $p^h$ at level $h+1$.
      By the parent property, $\dist(p,x)\le\dist(p^h,\parent(q^h))$.
      By the triangle inequality,
      \begin{align*}
        \dist(p^h,\parent(q^h))\le\dist(p^h,q^h)+\dist(q^h,\parent(q^h))\le c_r\tau^h+c_c\tau^{h+1}=(c_c+\frac{c_r}{\tau})\tau^{h+1}.
      \end{align*}
      Therefore, $\dist(p,x)<(c_c+c_r/\tau)\tau^{h+1}$, which implies that the covering constant of $T'$ is $c_c+\frac{c_r}{\tau}$.
      Note that the distance of any node in any level $\ell$ in $T'$ except $p^h$ to its parent is at most $c_c\tau^{\ell+1}$.

      Let $h$ be the minimum value so that $\dist(p,q)\le c_r\tau^h$.
      Then $c_r\tau^{h-1}<\dist(p,q)\le c_r\tau^h$, as such $h\coloneqq\lceil\log_\tau(\dist(p,q)/c_r)\rceil$.
      Insertion of $p$ at level $h$ should preserve the packing property, i.e. $\dist(p,q)>c_p\tau^h$.
      Since $c_p\le c_c$ and $c_r\ge \frac{2c_c\tau}{\tau-2}$, we have $c_p\tau^h\le c_c\tau^h<\frac{2c_c\tau}{\tau-2}\tau^h\le c_r\tau^h$.
      However, $c_p\tau^h<c_r\tau^{h-1}$ does not necessarily hold, so if $p$ is inserted at level $h$, it may violate the packing property.
      Furthermore, we have $\dist(p,q)>c_r\tau^{h-1}\ge \frac{2c_c\tau}{\tau-2}\tau^{h-1}>c_c\tau^{h-1}\ge c_p\tau^{h-1}$, which implies that the insertion of $p$ at level $h-1$ satisfies the packing property.
      Therefore, the insertion algorithm correctly maintains the packing property

      To prove the parent property, we need to show that the parent and children of $p^h$ in $T'$ are correct.
      Since $p^h\sim q^h$, Lemma~\ref{lem:functorial} implies $\parent(p^h)\sim\parent(q^h)$, so the algorithm correctly finds the parent of $p^h$.
      To show that the children of $p^h$ in $T'$ are correct, we first prove that $p^h$ cannot serve as the parent of any node with a level less than $h-1$, then we show that the algorithm correctly finds its children in level $h-1$.
      Consider a node $s^g\in T$, where $g\le h-2$.
      We have $\dist(p,\parent(s^{g}))>c_r\tau^{h-1}$, otherwise $\parent(s^{g})$ should have been the center of $p$.
      By the triangle inequality,
      \begin{align*}
        \dist(p^h,s^g)
        &\ge \dist(p^h,\parent(s^{g}))-\dist(\parent(s^{g}),s^g)
        >c_r\tau^{h-1}-c_c\tau^{g+1} \\
        &> \frac{2c_c\tau}{\tau-2}\tau^{h-1}-c_c\tau^{h-1}=c_c\frac{\tau+2}{\tau-2}\tau^{h-1}>c_c\tau^{h-1}.
      \end{align*}
      Therefore, $p$ cannot cover $s^g$, which implies that we only need to check the nodes at level $h-1$ to find children of $p^h$.
      Furthermore, we show that if $p^h$ is the closest node at level $h$ to a node $s^{h-1}$, then $p^h\sim\parent(s^{h-1})$, which implies that the algorithm correctly finds children of $p^h$.
      By the parent property, $\dist(p,s)<\dist(s^{h-1},\parent(s^{h-1}))$.
      By the triangle inequality,
      \begin{align*}
        \dist(p^h,\parent(s^{h-1}))
        &\le \dist(p^h,s^{h-1})+\dist(s^{h-1},\parent(s^{h-1}))
        <2\dist(s^{h-1},\parent(s^{h-1}))
        <2c_c\tau^h
        <c_r\tau^h.
      \end{align*}

      Lemma~\ref{lem:functorial} implies that the algorithm correctly finds the relatives of $p^h$ (nodes that have $p^h$ as their relative will be updated too).
      If $q^h$ is added to the tree, we do not need to find its relatives separately because $p^h$ is its only relative.
      Also, if $q^{h+1}$ is inserted to the tree, we do not need to update $\rel(q^{h+1})$ because it does not have any relatives other than itself.
      Therefore, the algorithm correctly updates relatives after each insertion.

      Eventually, $T'$ is semi-compressed because the algorithm removes those nodes that do not satisfy the semi-compressed condition (while updating children) and the created nodes have more than one relative or more than one child.
    \end{proof}


  \subsection{Bottom-Up Propagation} 
  \label{sub:bottom_up_propagation}

    If the insertion of a new point $p$ violates the local covering property (change the covering constant from $c_c$ to $c_c+c_r/\tau$), then the bottom-up propagation algorithm restores the covering property by promoting $p^\ell$ to higher levels of the tree as follows.
    Let $q^{\ell+1}\coloneqq\parent(p^\ell)$.
    First, we create node $p^{\ell+1}$ and make it as the parent of $p^\ell$.
    Then, we make the closest node among $\rel(\parent(q^{\ell+1}))$ to $p$ as the parent of $p^{\ell+1}$.
    Finally, we find relatives and children of $p^{\ell+1}$ in a way similar to the insertion algorithm (we also remove the nodes that do not satisfy the semi-compressed condition).
    If node $p^{\ell+1}$ still violates the covering property, we use the same procedure to promote it to a higher level.
    Here, we use iteration $i$ to indicate promotion of point $p$ to level $\ell+i$.

    \begin{lemma}
    \label{lem:dist_violatingnode_to_parent}
      Given $c_r\ge \frac{2c_c\tau}{\tau-2}$ and a violating node $p^\ell$, in the $i$-th iteration of the bottom-up propagation algorithm, $\dist(p^{\ell+i},\parent(p^{\ell+i}))\le (c_c+\frac{c_r}{\tau})\tau^{\ell+i+1}< c_r\tau^{\ell+i+1}$.
    \end{lemma}
    \begin{proof}
      We prove this lemma by induction.
      For the base case $i=0$, Lemma~\ref{lem:insert} implies $\dist(p,\parent(p^\ell))\le (c_c+\frac{c_r}{\tau})\tau^{\ell+1}$.
      Also, for $c_r\ge \frac{2c_c\tau}{\tau-2}$, $c_c+\frac{c_r}{\tau}\le \frac{c_r(\tau-2)}{2\tau}+\frac{c_r}{\tau}=\frac{c_r}{2}<c_r$.
      Assume that the lemma holds for some $i-1\ge 0$, and we show that it is also true for $i$.
      In other words, the distance between $p^{\ell+i-1}$ to $q^{\ell+i}\coloneqq\parent(p^{\ell+i-1})$ is greater than $c_c\tau^{\ell+i}$, as such $p$ should be promoted to level $\ell+i$.
      The algorithm finds the parent of $p^{\ell+i}$ among the relatives of $\parent(q^{\ell+i})$.
      Therefore, $\parent(p^{\ell+i})$ is a node in level $\ell+i+1$ so that $\dist(p^{\ell+i},\parent(p^{\ell+i}))\le \dist(p^{\ell+i},\parent(q^{\ell+i}))$.
      By the triangle inequality,
      \begin{align*}
        \dist(p^{\ell+i},\parent(p^{\ell+i}))
        &\le \dist(p^{\ell+i},\parent(q^{\ell+i}))
        \le \dist(p^{\ell+i},q^{\ell+i})+\dist(q^{\ell+i},\parent(q^{\ell+i}))\\
        &\le (c_c+\frac{c_r}{\tau})\tau^{\ell+i}+c_c\tau^{\ell+i+1}=(\frac{c_r}{\tau^2}+\frac{c_c}{\tau}+c_c)\tau^{\ell+i+1}\\
        &<(c_c+\frac{c_r}{\tau})\tau^{\ell+i+1}<c_r\tau^{\ell+i+1}. \qedhere
      \end{align*}
    \end{proof}

    The following lemma states that the bottom-up propagation algorithm correctly restores the covering property, and its proof is similar to the proof of Lemma~\ref{lem:insert}

    \begin{lemma}
    \label{lem:promotion}
      Given a semi-compressed tree $T\in\wnt(\tau, c_p, c_c+\frac{c_r}{\tau})$ with $c_r\ge\frac{2c_c\tau}{\tau-2}$.
      Let for all nodes $x^m\in T$ except $p^\ell$, $\dist(x^m,\parent(x^m))\le c_c\tau^{m+1}$ and for $p^\ell$, $c_c\tau^{\ell+1}<\dist(p^\ell,\parent(p^\ell))\le (c_c+\frac{c_r}{\tau})\tau^{\ell+1}$.
      Then, the bottom-up propagation algorithm results a semi-compressed tree $T'\in\wnt(\tau,c_p,c_c)$.
    \end{lemma}
    \begin{proof}
      First, we prove that the local packing, covering, and parent properties are mintained.
      Since $\dist(p^{\ell+i},\parent(p^{\ell+i}))>c_c\tau^{\ell+i+1}\ge c_p\tau^{\ell+i+1}$, the promotion does not modify the packing constant.
      Also, the violating node can be promoted up to the root (at level $+\infty$), so the algorithm results the covering constant of $c_c$.
      Using~\Cref{lem:dist_violatingnode_to_parent,lem:functorial}, $\parent(p^{\ell+i})\sim\parent(q^{\ell+i})$, so the parent is in $\rel(\parent(q^{\ell+i}))$.
      The proof of correctness of $\child(p^{\ell+i})$ is similar to Lemma~\ref{lem:insert}.

      It is easy to see that the relatives of $p^{\ell+i}$ are among $\rel(\parent(q^{\ell+i}))$ and the algorithm correctly finds the relatives.
      Finally, we need to show that $T'$ is semi-compressed.
      In other words, we should prove that all the created nodes for $p$ are required in $T'$.
      Lemma~\ref{lem:dist_violatingnode_to_parent} implies that node $p^{\ell+i}$ has at least one relative besides itself, i.e. $q^{\ell+i}$, which is the old parent of $p^{\ell+i-1}$ in iteration $i-1$.
      So, it is always necessary to create node $p^{\ell+i}$ in the $i$-th iteration.
    \end{proof}


  \subsection{Analysis} 
  \label{sub:analysis}

    In the following theorem, we analyze the running time of the bottom-up construction algorithm without considering the point location cost which will be handled in Section~\ref{sec:point_location}.

    \begin{theorem}
    \label{thm:construction_without_PL}
      Not counting the PL step, the bottom-up construction runs in $O(\rho^{O(1)}n)$ time.
    \end{theorem}
    \begin{proof}
      To prove this theorem, we use an amortized analysis which imposes the cost of each iteration on a node in the output.
      In the promotion phase, Lemma~\ref{lem:dist_violatingnode_to_parent} implies that every node of $p^{\ell+i}$ has at least one relative besides itself, namely $q^{\ell+i}$.
      So, we can make $q^{\ell+i}$ responsible to pay the cost of iteration $i$ for $p$.
      Note that a node $q^{\ell+i}$ will not be removed by any points that will be processed next, because $p^{\ell+i}\sim q^{\ell+i}$ satisfies the semi-compressed condition.
      In other words, there is always a node in the output that pays the cost of promotion.
      By Lemma~\ref{lem:ntsize}, the cost of each iteration is $\rho^{O(1)}$ and $q^{\ell+i}$ has $\rho^{O(1)}$ relatives, as such $q^{\ell+i}$ receives $\rho^{O(1)}$ cost in total.
      Therefore, to pay the cost of all promotions for all $n$ points, each node in the output requires $\rho^{O(1)}$ charge.
      By Lemma~\ref{lem:promotion}, the output is semi-compressed and Theorem~\ref{thm:semi-compressed_linear_size} implies that it has $O(\rho^{O(1)}n)$ size.
      Thus, the total cost of all promotions for all $n$ points does not exceed $O(\rho^{O(1)}n)$.

      Notice that when a point is inserted to the tree for the first time, it does not necessarily have any other relatives.
      However, the insertion occurs only once for each point and it requires $\rho^{O(1)}$ time.
      Therefore, all insertions can be done in $O(\rho^{O(1)}n)$ time.
    \end{proof}


  \section{Randomized Incremental Construction} 
\label{sec:point_location}
  In this section, we show how to eagerly compute the centers of all uninserted points.
  The centers are updated each time either a new node is added or an existing node is deleted by doing a local search among parents, children, and relatives of the node.
  We show that the following invariant is satisfied after each insertion or deletion.

  \begin{invariant*}
    The centers of all uninserted points are correctly maintained.
  \end{invariant*}

  In Section~\ref{sec:pl_algorithm}, we present the point location algorithm.
  Then, in Section~\ref{sec:pl_events}, we show that for a random ordering of points, the point location takes $O(n\log n)$ time in expectation.
  As this point location work is the main bottleneck in the algorithm, the following theorem is main contribution of this paper.

  \begin{theorem}
    Given a random permutation $\pi=\langle p_1,\ldots,p_n\rangle$.
    A net-tree $T\in \nt(\tau,\frac{c_p(\tau-1)-2c_c}{2(\tau-1)},\frac{c_c\tau}{\tau-1})$ with $c_r=\frac{2c_c\tau}{\tau-4}$ can be constructed from $\pi$ in $O(\rho^{O(1)}n\log n)$ expected time, where $\tau\ge \max\{5,\frac{2c_c}{c_p}+2\}$ and $0<c_p\le c_c<\frac{c_p(\tau-1)}{2}$ are constants.
  \end{theorem}

  \subsection{The Point Location Algorithm} 
\label{sec:pl_algorithm}

  We will describe a simple point location algorithm referred to as \emph{the PL algorithm} from here on.
  The idea is to store the center of each uninserted point, and for each node, a list of uninserted points whose center is that node (i.e.\ the Voronoi cell of the node).
  Formally, the \emph{cell} of a node $p^\ell$, denoted $\cell(p^\ell,T)$, is the list of points in $\vor(p^\ell)$.
  We partition the points $x$ of $\cell(p^\ell,T)$ into $\cell_{in}(p^\ell,T)$ and $\cell_{out}(p^\ell,T)$ depending on whether $\dist(p,q)\le c_p\tau^{\ell-1}/2$ or not.
  This separation saves some unnecessary distance computations.

  Each time a node is added to the tree $T$ to create a new tree $T'$, we update the centers and cells nearby.
  There are two different ways that a new node is created, either it splits a jump or it is inserted as a child of an existing node.
  If a jump from $p^h$ to $p^g$ is split at level $\ell$, then we select $\cell(p^\ell,T')$ from the nodes of $\cell(p^h,T)$.
  If $p^\ell$ is inserted as a child of $s^{\ell+1}$, then we select $\cell(p^\ell,T')$ from $\{\cell_{out}(x^h,T)\mid x^h\in\rel(s^{\ell+1})\cup\child(\rel(s^{\ell+1}))\cup\child(\rel(p^\ell))\}$.
  A node $p^\ell$ with parent $p^h$ may be removed if required by the compression.
  In such cases, $\cell_{in}(p^\ell,T)$ is added to $\cell_{in}(p^h,T')$ and the points in $\cell_{out}(p^\ell,T)$ are tested to determine which points belong to $\cell_{in}(p^\ell,T')$ or $\cell_{out}(p^\ell,T')$.

  The following lemma shows that the PL algorithm correctly maintains the invariant.
  \begin{lemma}\label{lem:pl_correct}
    The PL algorithm correctly maintains the invariant after the insertion or deletion of a node.
  \end{lemma}
  \begin{proof}
    First, we prove that after deletion of a node $p^\ell$, the PL algorithm correctly updates the center of every uninserted point $q\in \cell(p^\ell,T)$.
    Let $p^h$ be the parent of $p^\ell$.
    Note that by the definition of a center, $p^h$ should be the new center for $q$ after deletion of $p^\ell$.
    If $q\in\cell_{in}(p^\ell,T)$, then $\dist(p,q)\le c_p\tau^{\ell-1}/2<c_p\tau^{h-1}/2$, which means $q\in\cell_{in}(p^h,T')$.
    Otherwise, $q$ can belong to the inner or the outer cell of $p^h$.

    Now, we prove that if $p^\ell$ is added and $q\in\cell(p^\ell,T')$, then $q$ belongs to the set of nearby uninserted points of $p^\ell$.
    \begin{enumerate}
      \item[(a)] \emph{$p^\ell$ splits a jump from $p^h$ to $p^g$}:
      Since $q\in\cell(p^{\ell},T')$ and $T'$ has only one node $p^\ell$ more than $T$, $q$ should have been in a cell of a node of $p$.
      Also, $\dist(p,q)>c_r\tau^g$, because otherwise $p^g$ should be the center of $q$ in $T$.
      So, $\dist(p,q)\le c_r\tau^\ell<c_r\tau^h$, which means $q\in \cell(p^h,T)$.
      \item[(b)] \emph{$p^\ell$ is inserted as a child of $s^{\ell+1}$}:
      \label{prf:pl_correctb}
      First we show that $\dist(s,q)\le c_r\tau^{\ell+1}$.
      From Lemma~\ref{lem:dist_violatingnode_to_parent}, $\dist(p,s)\le (c_c+c_r/\tau)\tau^{\ell+1}$.
      By the triangle inequality, $\dist(s,q)\le \dist(s,p)+\dist(p,q)\le (c_c+2c_r/\tau)\tau^{\ell+1}$.
      For $c_r\ge c_c\tau/(\tau-2)$, $\dist(s,q)\le c_r\tau^{\ell+1}$.
      So, there exists at least one node in level $\ell+1$ that can be served as the center of $q$ before the insertion of $p^\ell$ and it is $s^{\ell+1}$.
      However, $q$ might be closer to any other nodes, so $q$ is not necessarily in $\cell(s^{\ell+1},T)$.

      If $q\in \cell(x^h,T)$, then we show that $\ell-1\le h\le\ell+1$.
      Suppose for contradiction that $h<\ell-1$.
      Then,
      \begin{align*}
        \dist(p,x)
        &\le\dist(p,q)+\dist(q,x) \because{by the triangle inequality} \\
        &<\dist(q,x)+\dist(q,x) \because{$\dist(p,q)<\dist(q,x)$ because $q\in \cell(p^\ell,T')$} \\
        &\le 2c_r\tau^h \because{$q\in \cell(x^h,T)$} \\
        &\le c_r\tau^{\ell-1} \because{$h\le\ell-2$ and $\tau\ge 2$}
      \end{align*}
      Therefore, $p^{\ell-1}\sim x^{\ell-1}$ and $q\in\cell(p^{\ell-1},T)$, which is a contradiction because $p$ is inserted at level $\ell$.

      Suppose for contradiction that $h>\ell+1$.
      Then,
      \begin{align*}
        \dist(x,s)
        &\le\dist(x,q)+\dist(q,s) \because{by the triangle inequality} \\
        &<\dist(q,s)+\dist(q,s) \because{$\dist(x,q)<\dist(q,s)$ because $q\in \cell(x^h,T)$} \\
        &\le 2(\dist(q,p)+\dist(p,s)) \because{by the triangle inequality} \\
        &\le 2(c_r/\tau+c_c+c_r/\tau)\tau^{\ell+1} \because{by Lemma~\ref{lem:dist_violatingnode_to_parent}, $\dist(p,s)\le(c_c+c_r/\tau)\tau^{\ell+1}$} \\
        &\le c_r\tau^{\ell+1} \because{for $c_r\ge 2c_c\tau/(\tau-4)$}
      \end{align*}
      So, $s^{\ell+1}\sim x^{\ell+1}$.
      Also, $\dist(q,x)<\dist(q,s)\le c_r\tau^{\ell+1}$.
      Therefore, $q\in\cell(x^{\ell+1},T)$, which is a contradiction.
      It is easy to see that $x^h$ belongs to $\rel(s^{\ell+1})\cup\child(\rel(s^{\ell+1}))\cup\child(\rel(p^\ell))$.

      Finally, we prove that the points in cell $\cell(p^\ell,T')$ are in the outer cells of the nearby nodes.
      For contradiction, suppose that $q\in \cell_{in}(x^h,T)$, where $\ell-1\le h\le\ell+1$.
      So, $\dist(q,x)\le c_p\tau^{h-1}/2$.
      Then, $\dist(p,q)<\dist(q,x)$ and by the triangle inequality, $\dist(p,x)\le\dist(p,q)+\dist(q,x)<2\dist(q,x)\le c_p\tau^{h-1}$.
      If $\ell\le h\le\ell+1$, then $\dist(p,x)\le c_p\tau^\ell$ and it contradicts with the packing property at level $\ell$.
      Otherwise, if $h=\ell-1$, then $\dist(p,x)\le c_p\tau^{\ell-2}$ and it also contradicts with the packing property at level $\ell-1$.\qedhere
    \end{enumerate}
  \end{proof}


  \subsection{Analysis of the PL Algorithm} 
\label{sec:pl_events}
  When a node of $p$ checks an uninserted point $q$ to see if $q$ belongs to its cell, we say $p$ \emph{touches} $q$.
  To analyze the point location algorithm, we should count the total number of touches, because each touch corresponds to a distance computation.
  Note that a point does not change its center each time it is touched.
  This is the main challenge in point location, to avoid touching a point too many times unnecessarily.

  We classify the touches into three groups of \emph{basic touches}, \emph{split touches}, and \emph{merge touches}.
  Then, we use a backwards analysis to bound the expected number of such touches.
  The standard approach of using backwards analysis for randomized incremental constructions will not work directly for the tree construction, because the structure of the tree is highly dependent on the order the points were added.
  Instead, we define random events that can happen for each point $p_i$ of a permutation $\langle p_1,\ldots ,p_n\rangle$ in $P_j\coloneqq\{p_1,\ldots,p_j\}$, where $j<i$.
  These events are defined only in terms of the points in the permutation, and do not depend on a specific tree.
  We show that each point is involved in $O(\log n)$ such events.
  Later, we show that the touches all correspond to these random events.

  \begin{figure}[!tbh]
    \centering
    \begin{subfigure}[b]{0.47\columnwidth}
      \includegraphics[width=\textwidth]{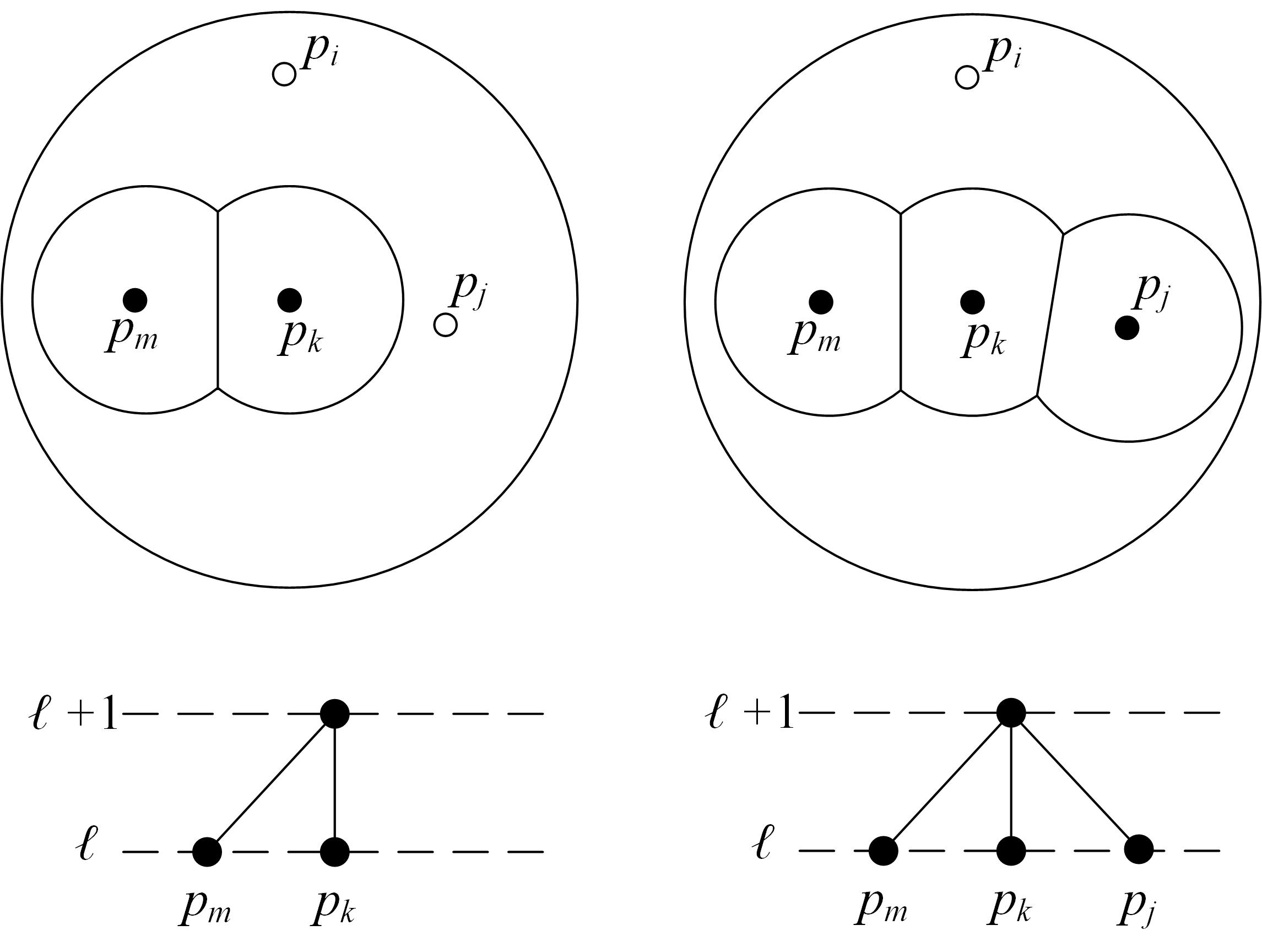}
      \centering
      \caption{}
      \label{fig:basictouch}
    \end{subfigure}
    \quad\quad
    \begin{subfigure}[b]{0.47\columnwidth}
      \includegraphics[width=\textwidth]{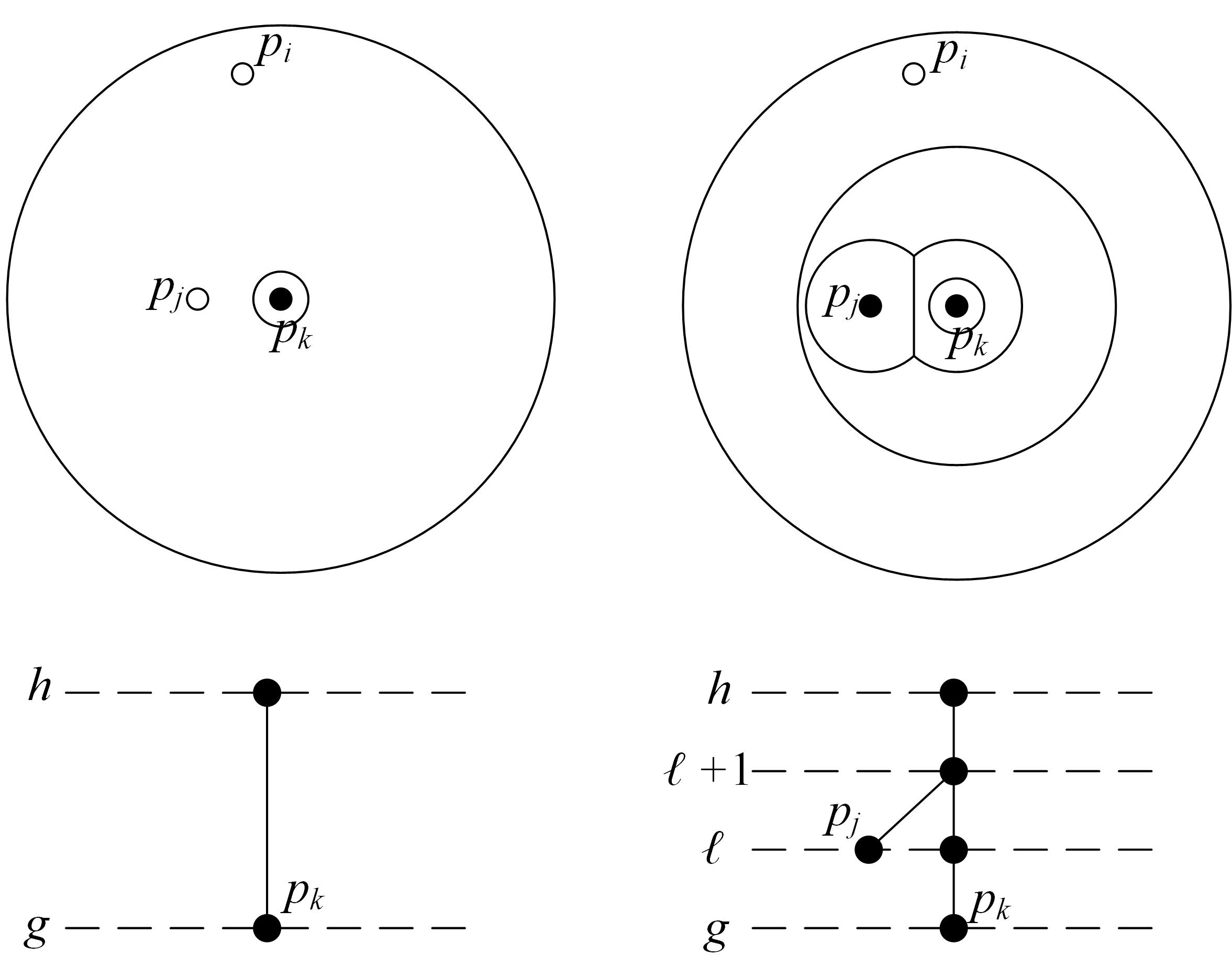}
      \centering
      \caption{}
      \label{fig:splittouch}
    \end{subfigure}
    \caption{The approximate Voronoi diagrams at the top are induced on the part of net-trees at the bottom.
    White dots show the uninserted points.
    (a) The insertion of $p_j$ at level $\ell$ results a basic touch from $p_j$ on $p_i$.
    Before the insertion, $p_i$ and $p_j$ belong to cell $\cell(p_k^{\ell+1},T)$, and after that, $p_i$ remains in the same cell.
    (b) The insertion of $p_j$ at level $\ell$ results a split touch from $p_k$ on $p_i$.
    Before the insertion, $p_i$ and $p_j$ belong to cell $\cell(p_k^{h},T)$, and after that, $p_i$ remains in the same cell.}
    \label{fig:touch}
  \end{figure}

  If $p_i$ is touched by a new point $p_j$, then we say a \emph{basic touch} has happened, see Fig.~\ref{fig:basictouch}.
  If $p_i$ is touched by the point of $\centersite(p_i)$ after the insertion of $p_j$, then a \emph{split touch} has happened, see Fig.~\ref{fig:splittouch}.
  Intuitively, a split touch in the tree occurs when $\centersite(p_i)$ is the top of a jump and the insertion of $p_j$ results that jump to be split at a lower level.
  By the PL algorithm, the cell of a new node can be found from the cell of its parent.
  Therefore, $p_i$ and all other points in the cell of $\centersite(p_i)$ will be touched by the point of $\centersite(p_i)$ at a smaller scale.
  A split touch is either below or above, which will be discussed later.
  Similarly, If $p_i$ is touched by the point of $\centersite(p_i)$ after the deletion of $\centersite(p_i)$ triggered by the insertion of $p_j$, then a \emph{merge touch} has happened.
  In other words, a merge touch occurs if the insertion of $p_j$ results $\centersite(p_i)$ to be deleted and its adjacent edges merged to a jump.
  In this case, the PL algorithm moves $p_i$ and other points in the cell of $\centersite(p_i)$ to the cell of the parent of $\centersite(p_i)$.
  For the sake of simplicity, we abuse the notion of touches for split and merge cases in the following way.
  If in a split or merge touch, the point of $\centersite(p_i)$ touches $p_i$, then we charge $p_j$ for that touch and we say that $p_j$ touches $p_i$.

  \begin{lemma}
  \label{lem:const_touch}
    A point $p_j$ can touch $p_i$ at most $\rho^{O(1)}$ times.
  \end{lemma}
  \begin{proof}
    First we count the number of basic touches.
    Note that when we promote $p_j$ to a higher level, $p_j$ might touch $p_i$ more than once.
    From the algorithm in Section~\ref{sec:pl_algorithm}, the promoting node only checks nearby cells from one level down to one level up.
    Therefore, $p_j$ can only touch $p_i$ at most three times.

    Now, we compute the number of split touches.
    When $p_j$ splits a jump on $\centersite(p_i)$, it may create two new nodes for $\centersite(p_i)$, see Fig.~\ref{fig:splittouch}.
    So, $p_i$ can be touched by $p_j$ at most twice.
    If $p_j$ requires to be promoted to a higher level, $p_i$ may receive more touches.
    This case only happens when $p_i$ is touched, but its center remains unchanged.
    Let $p_j$ be inserted at level $\ell$ and $q^{\ell+1}\coloneqq\parent(p_j^\ell)$.
    From Lemma~\ref{lem:dist_violatingnode_to_parent}, $\dist(p_j,q)\le c_r\tau^\ell+c_c\tau^{\ell+1}$.
    In the following, we will show that in the promotion process, $q$ cannot touch $p_i$ more than $\log_\tau (c_r/c_c+\tau)$ times.
    To prove this bound, we show that the promotion cannot continue more than $k>1$ levels above $\ell$ with the same parent $q$.
    In other words, $\dist(p_j,q)\le c_c\tau^{\ell+k}$, which satisfies the covering property.
    So $c_r\tau^\ell+c_c\tau^{\ell+1}\le c_c\tau^{\ell+k}$, which results $k\ge\lceil\log_\tau (c_r/c_c+\tau)\rceil$.
    Therefore, the total number of split touches from $p_j$ on $p_i$ is also constant.

    Finally, we prove that the number of merge touches is also constant.
    Recall that when a node is deleted from the tree, the PL algorithm only checks the uninserted points in its outer cell to determine which points should be moved to the inner or outer cells of its parent.
    Let $p_m^{\ell}$ be the node to be deleted and $p_i\in\cell(p_m^\ell,T)$.
    Here, we wish to find a level $\ell+k$, where $k>0$, such that $p_i$ goes to the inner cell of $p_m^{\ell+k}$ and as such does not recieve more merge touches from $p_j$.
    Therefore, $\dist(p_i,p_m)\le c_p\tau^{\ell+k-1}/2$.
    By the definition of a center, $p_i$ cannot be in a distance farther that $c_r\tau^\ell$ from $p_m$.
    So, we have $c_r\tau^{\ell}\le c_p\tau^{\ell+k-1}/2$, which results $k\ge 1+\log_\tau(2c_r/c_p)$.
    Therefore, $p_j$ can only touch $p_i$ a constant number of times via merge touches.
  \end{proof}

  In this section, our goal is finding the expected number of touches for each uninserted point in a random permutation.
  In Section~\ref{ssub:basic_touches}, we show that the expected number of basic touches is $O(n\log n)$.
  In Section~\ref{ssub:split_and_merge_touches}, we prove that the expected number of split touches is $O(n\log n)$, and then we show that the number of merge touches is bounded by the number of split touches.
  The following theorem states the total cost of point location.

  \begin{theorem}
    \label{thm:pl_time}
    The expected running time of point location in the randomized incremental construction algorithm is $O(\rho^{O(1)}n\log n)$.
  \end{theorem}

  \subsubsection{Basic Touches} 
  \label{ssub:basic_touches}
    In this section, we first prove that the distance of every point touching $p_i$ with a basic touch is bounded by the distance of $p_i$ to its nearest neighbor among the inserted points.
    Using this observation, we divide a permutation into phases, where each phase is an interval in which the nearest neighbor of $p_i$ remains unchanged, see Fig.~\ref{fig:phase_basic_touch}.
    We show that the number of basic touches on $p_i$ in each phase is constant.
    Then, using a backwards analysis we show that the expected number of phases for each point in $O(\log n)$.
    Therefore, the expected number of basic touches for all points is $O(n\log n)$.

    \begin{figure}
      \centering
      \includegraphics[width=0.60\columnwidth]{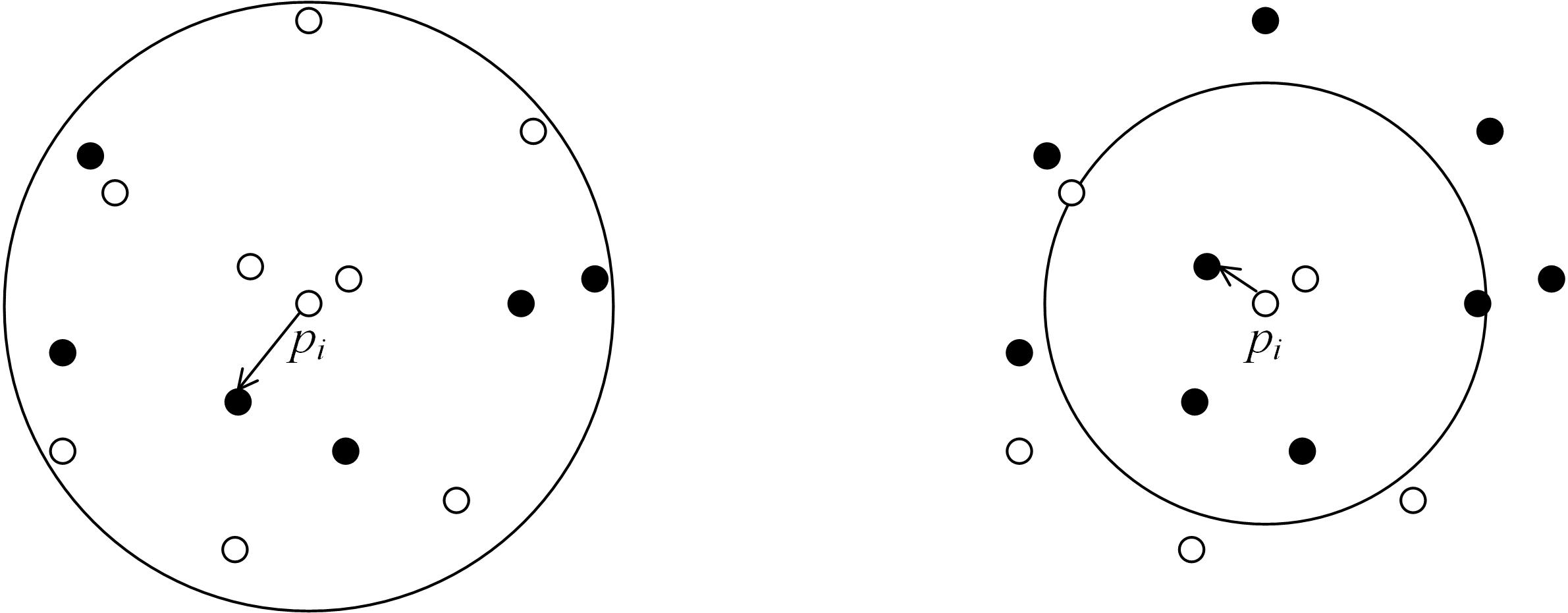}
      \caption{
      The solid and white dots show the inserted and uninserted points, respectively.
      An arrow shows the nearest neighbor of $p_i$ among the inserted points.
      The uninserted points in each ball are the only points that may touch $p_i$ via a basic touch.
      With the change of the nearest neighbor, the ball containing touching points shrinks.
      The points that get inserted from the the left figure to the right figure belong to the same phase.}
      \label{fig:phase_basic_touch}
    \end{figure}

    \begin{lemma}
    \label{lem:NN_touching_point}
      If $p_j$ touches $p_i$ via a basic touch, then $\dist(p_i,p_j)<\frac{24 c_c\tau^4(\tau-1)}{c_p(\tau+2)(\tau-4)}\dist(p_i,P_{j-1})$
    \end{lemma}
    \begin{proof}
      Let $p_k^\ell$ be the center of $p_i$ in $P_{j-1}$, where $k<j<i$.
      According to the PL algorithm, $p_j$ can be in any level between $\ell-1$ and $\ell+1$.
      If $p_j$ is at level $\ell+1$, then $p_k\in\child(\rel(p_j^{\ell+1}))$.
      So, by the triangle inequality,
      \[\dist(p_j,p_k)\le \dist(p_j^{\ell+1},\parent(p_k^\ell))+\dist(\parent(p_k^\ell),p_k^\ell)\le c_r\tau^{\ell+1}+c_c\tau^{\ell+1}<2c_r\tau^{\ell+1}.\]
      If $p_j$ is at level $\ell-1$, then $p_k\in\rel(\parent(p_j^{\ell-1}))$.
      Using Lemma~\ref{lem:dist_violatingnode_to_parent} and the triangle inequality,
      \[\dist(p_j,p_k)\le \dist(p_j^{\ell-1},\parent(p_j^{\ell-1}))+\dist(\parent(p_j^{\ell-1}),p_k^\ell) \le (c_c+c_r/\tau)\tau^{\ell}+c_r\tau^\ell<2c_r\tau^\ell.\]
      If $p_j$ is at level $\ell$, then by the triangle inequality and Lemma~\ref{lem:dist_violatingnode_to_parent} we have
      \begin{align*}
        \dist(p_j,p_k)
        &\le \dist(p_j^\ell,\parent(p_j^\ell))+\dist(\parent(p_j^\ell),\parent(p_k^\ell))+\dist(\parent(p_k^\ell),p_k^\ell)\\
        &\le(c_c+c_r/\tau)\tau^{\ell+1}+c_r\tau^{\ell+1}+c_c\tau^{\ell+1}
        <2c_r\tau^{\ell+1}. \because{$c_r=2c_c\tau/(\tau-4)$}
      \end{align*}

      Therefore, we have $\dist(p_j,p_k)<2c_r\tau^{\ell+1}$ for all cases, as such $\dist(p_i,p_j)\le \dist(p_i,p_k)+\dist(p_k,p_j)\le c_r\tau^\ell+2c_r\tau^{\ell+1}<3c_r\tau^{\ell+1}$.
      As we saw earlier in Section~\ref{sec:pl_algorithm}, $p_j$ touches $p_i$ with a basic touch if $p_i$ is in the outer cell of $p_k$, which means $\dist(p_i,p_k)>c_p\tau^{\ell-1}/2$.
      Combining the last two inequalities results $\dist(p_i,p_j)<6 c_r\tau^2\dist(p_i,p_k)/c_p$.
      From Lemma~\ref{lem:NN_center}, for $c_r=2c_c\tau/(\tau-4)$, we have $\dist(p_i,p_k)<2\tau(\tau-1)/(\tau+2)\dist(p_i,P_{j-1})$.
      The lemma follows from the last two inequalities.
    \end{proof}

    \begin{lemma}
    \label{lem:const_basic_touch}
      If $\dist(p_i,P_k)=\dist(p_i,P_j)$ for $k<j<i$, then the number of basic touches on $p_i$ from $p_k$ to $p_j$ is $\rho^{O(1)}$.
    \end{lemma}
    \begin{proof}
      Let $q$ be the closest point to $p_i$ in both $P_k$ and $P_j$.
      Also, let $h\coloneqq\lceil\log_\tau (\dist(p_i,q)/c_r)\rceil$.
      If two points $x,y\in\{p_k,\ldots,p_j\}$ touch $p_i$ at levels $f$ and $g$ and the minimum distance remains unchanged, then $f,g\ge h-1$ because the PL algorithm in Section~\ref{sec:pl_algorithm} checks the nearby cells from one level down to one level up.
      By the packing property, $\dist(x,y)>c_p\tau^{h-1}$.
      By definition, $\dist(q,p_i)\le c_r\tau^h$.
      The last two inequalities result $\dist(x,y)>\frac{c_p}{c_r\tau}\dist(q,p_i)$.
      By Lemma~\ref{lem:NN_touching_point}, every touching point in $\{p_k,\ldots,p_j\}$ is within distance $\frac{24 c_c\tau^4(\tau-1)}{c_p(\tau+2)(\tau-4)}\dist(p_i,P_j)$ from $p_i$.
      Using the last two inequalities and the Packing Lemma, there are a constant number of points from $p_k$ to $p_j$ that can touch $p_i$ but not changing the minimum distance from $p_i$.
    \end{proof}

    \begin{theorem}
    \label{thm:expected_basic_touches}
      The expected number of basic touches in a random permutation is $O(\rho^{O(1)}n\log n)$.
    \end{theorem}
    \begin{proof}
      Using Lemma~\ref{lem:const_basic_touch}, only a constant number of basic touches on a point $p_i$ can occur before the distance from $p_i$ to the inserted points must go down.
      Therefore, it suffices to bound $E\big[|\{j\mid\dist(p_i,P_{j-1})\neq\dist(p_i,P_j)\}|\big]$.
      We observe that $\dist(p_i,P_{j-1})\neq\dist(p_i,P_j)$ only if $p_j$ is the unique nearest neighbor of $p_i$ in $P_j$.
      Using a standard backwards analysis, this event occurs with probability $1/j$.
      Therefore, $E\big[|\{j\mid\dist(p_i,P_{j-1})\neq\dist(p_i,P_j)\}|\big]\le\sum_{j=1}^{n}1/j=O(\log n)$.
      By Lemma~\ref{lem:const_touch}, each point $p_j$ may cause a constant number of basic touches on $p_i$.
      So, the expected number of basic touches on $p_i$ is $O(\rho^{O(1)}\log n)$, which results $O(\rho^{O(1)}n\log n)$ touches in expectation for the permutation.
    \end{proof}


  \subsubsection{Split and Merge Touches} 
  \label{ssub:split_and_merge_touches}
    In this section, we define bunches near a point $p_i$.
    These bunches are sufficently-separated disjoint groups of points around $p_i$.
    We show that each point has a constant number of bunches nearby.
    Then, we define two random events for each point based on its nearby bunches.
    We prove that the expected number of events for each point of a permutation is $O(\log n)$.
    Then, we show that the number of split touches can be counted by such events.
    Finally, we prove that the number of merge touches can be bounded in terms of the number of split touches.

    \begin{definition}
    \label{def:bunch}
      $\bunch\subseteq P_j$ is a \emph{bunch} near $p_i$, if there exists a center $x\in \bunch$ such that
      \begin{enumerate}
        \item \label{def:bunch1} $\bunch=\ball(x,\alpha\dist(p_i,x))\cap P_j$,
        \item \label{def:bunch2} $\big[\ball(x,\beta\dist(p_i,x))\setminus\ball(x,\alpha\dist(p_i,x))\big]\cap P_j=\emptyset$,
        \item \label{def:bunch3} $\dist(p_i,x)\le \frac{2\tau(\tau-1)}{\tau+2}\dist(p_i,P_j)$,
      \end{enumerate}
      where $j<i$, $0<\alpha\le 0.5$ and $\beta\ge 2\alpha$.
    \end{definition}
    \begin{figure}
      \centering
      \includegraphics[width=0.37\columnwidth]{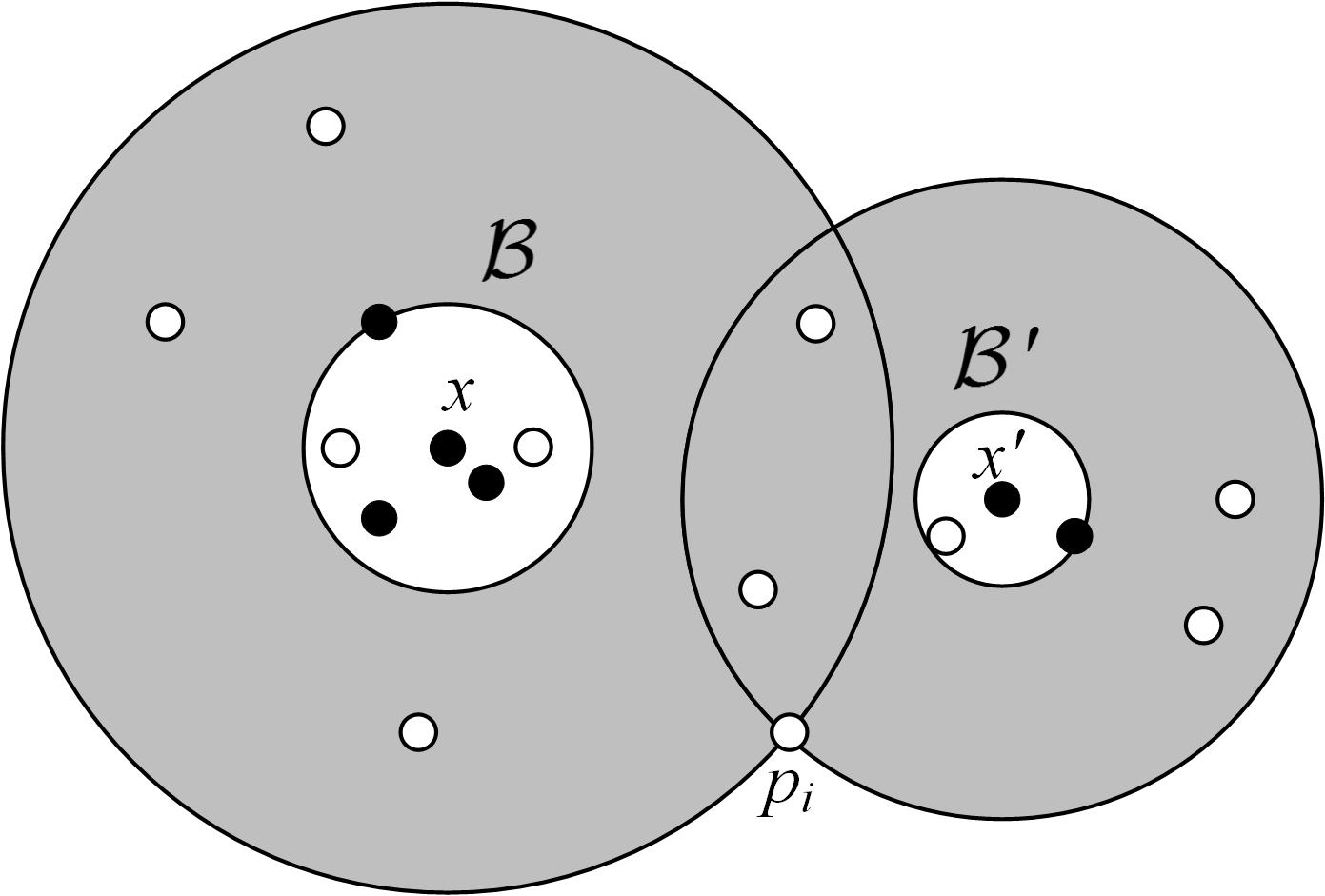}
      \caption{
      The solid dots are points of $P_j$ and white dots are points in $P\setminus P_j$.
      $\bunch$ and $\bunch'$ are two bunches near $p_i$ centered at $x$ and $x'$.
      The smaller balls contain the points in the corresponding bunches.
      No points of $P_j$ lies in the shaded region.}
      \label{fig:bunch}
    \end{figure}

    See Fig.~\ref{fig:bunch} for an illustration of bunches.
    Note that the third property is the result of Lemma~\ref{lem:NN_center} and $c_r=\frac{2c_c\tau}{\tau-4}$.
    Next, we show that there is a constant number of bunches near each point in a permutation.

    \begin{lemma}
    \label{lem:bunch_min_dist}
      Let $\bunch$ and $\bunch'$ be two distinct bunches with centers $x$ and $x'$ near $p_i$ in $P_j$, respectively.
      Then $\dist(x,x')>\frac{\alpha}{\alpha+1}\max\{\dist(p_i,x),\dist(p_i,x')\}$.
    \end{lemma}
    \begin{proof}
      Without loss of generality, let $\dist(p_i,x)>\dist(p_i,x')$.
      Suppose for contradiction, $\dist(x,x')\le\frac{\alpha}{\alpha+1}\dist(p_i,x)<\alpha\dist(p_i,x)$.
      The first property of a bunch results $x'\in\bunch$.
      Let $y'\in \bunch'\setminus\bunch$.
      The second property of $\bunch$ results $\dist(x,y')>\beta\dist(p_i,x)$.
      Using the triangle inequality,
      \begin{equation}
      \label{eq:bunch_min_dist1}
        \dist(x',y')\ge\dist(y',x)-\dist(x,x')>\beta\dist(p_i,x)-\frac{\alpha}{\alpha+1}\dist(p_i,x)=(\beta-\frac{\alpha}{\alpha+1})\dist(p_i,x).
      \end{equation}
      Also, using the first property of $\bunch'$ and the triangle inequality,
      \begin{equation}
      \label{eq:bunch_min_dist2}
        \dist(x',y')\le\alpha\dist(p_i,x')\le\alpha(\dist(p_i,x)+\dist(x,x'))\le\alpha(1+\frac{\alpha}{\alpha+1})\dist(p_i,x)
      \end{equation}
      By (\ref{eq:bunch_min_dist1}) and (\ref{eq:bunch_min_dist2}), $(\beta-\frac{\alpha}{\alpha+1})\dist(p_i,x)<\alpha(1+\frac{\alpha}{\alpha+1})\dist(p_i,x)$.
      Therefore, $\beta<\frac{\alpha}{\alpha+1}+\frac{2\alpha^2+\alpha}{\alpha+1}=2\alpha$, which is a contradiction because by definition $\beta\ge 2\alpha$.
      Thus, $\dist(x,x')>\frac{\alpha}{\alpha+1}\dist(p_i,x)$, as required.
    \end{proof}

    \begin{lemma}
    \label{lem:constant_bunches}
      For some constants $\alpha$ and $\beta$, there are $\rho^{O(1)}$ bunches near $p_i$.
    \end{lemma}
    \begin{proof}
      From Lemma~\ref{lem:bunch_min_dist}, for any bunches $\bunch$ and $\bunch'$ near $p_i$ with centers $x$ and $x'$, we have $\dist(x,x')>\frac{\alpha}{\alpha+1}\max\{\dist(p_i,x),\dist(p_i,x')\}\ge\frac{\alpha}{\alpha+1}\dist(p_i,P_j)$.
      So, the third property of bunches and the Packing Lemma imply the bound.
    \end{proof}

    The following lemma shows that each jump in a local net-tree corresponds to an empty annulus around the corresponding point.
    We will use this lemma to show the relation between jumps in a net-tree and bunches in a point set.

    \begin{lemma}
    \label{lem:empty_annulus}
      In a semi-compressed local net-tree $T\in\wnt(\tau ,c_p,c_c)$ with $c_r=\frac{2c_c\tau}{\tau-4}$, if there is a jump from $p^\ell$ to $p^h$, then $P_{p^h}\subset \ball(p,\frac{1}{2}c_r\tau^{h})$ and $\ball(p,\frac{1}{2}c_r\tau^{\ell-1})\setminus\ball(p,\frac{1}{2}c_r\tau^{h})=\emptyset$.
    \end{lemma}
    \begin{proof}
      Using Lemma~\ref{lem:covering}, $P_{p^h}\subseteq\ball(p,\frac{c_c\tau}{\tau-1}\tau^{h})=\ball(p,\frac{c_r(\tau-4)}{2(\tau-1)}\tau^h)\subset\ball(p,\frac{1}{2}c_r\tau^{h})$.
      Now, let $q\notin P_{p^\ell}$ and $x^{\ell-1}$ be the ancestor of $q$ at level $\ell-1$.
      Since $p^{\ell-1}\notin T$ and $T$ is semi-compressed, $\dist(p,x)>c_r\tau^{\ell-1}$.
      By the triangle inequality,
      \begin{align*}
        \dist(p,q)\ge\dist(p,x)-\dist(x,q)>c_r\tau^{\ell-1}-\frac{c_c\tau}{\tau-1}\tau^{\ell-1}>c_r\tau^{\ell-1}-\frac{c_r(\tau-4)}{2(\tau-1)}\tau^{\ell-1}>\frac{1}{2}c_r\tau^{\ell-1}.
      \end{align*}
      Notice that if $q$ does not have an ancestor in level $\ell-1$, then the radius of the outer ball becomes even larger because $\dist(x,q)\le \frac{c_c\tau}{\tau-1}\tau^{\ell-2}$.
      See Fig.~\ref{fig:fig4}.
    \end{proof}

    \begin{figure}[!tbh]
      \centering
      \includegraphics[width=0.25\columnwidth]{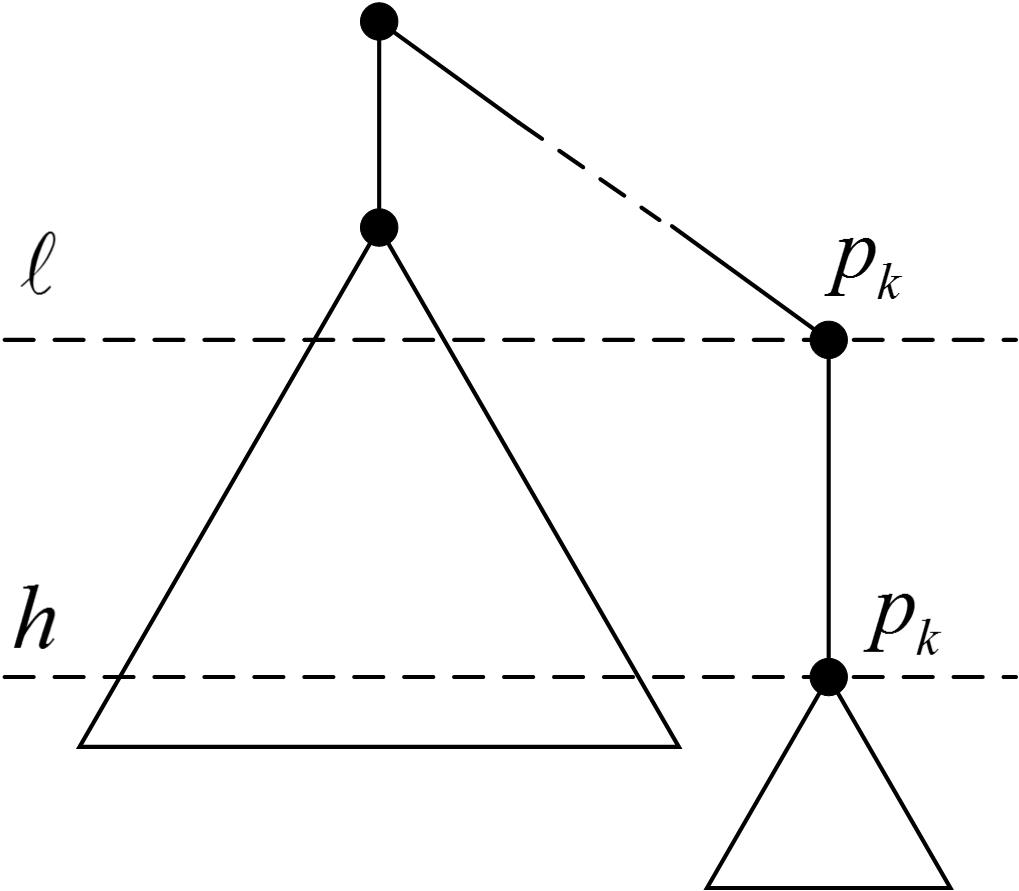}
      ~~~
      \includegraphics[width=0.25\columnwidth]{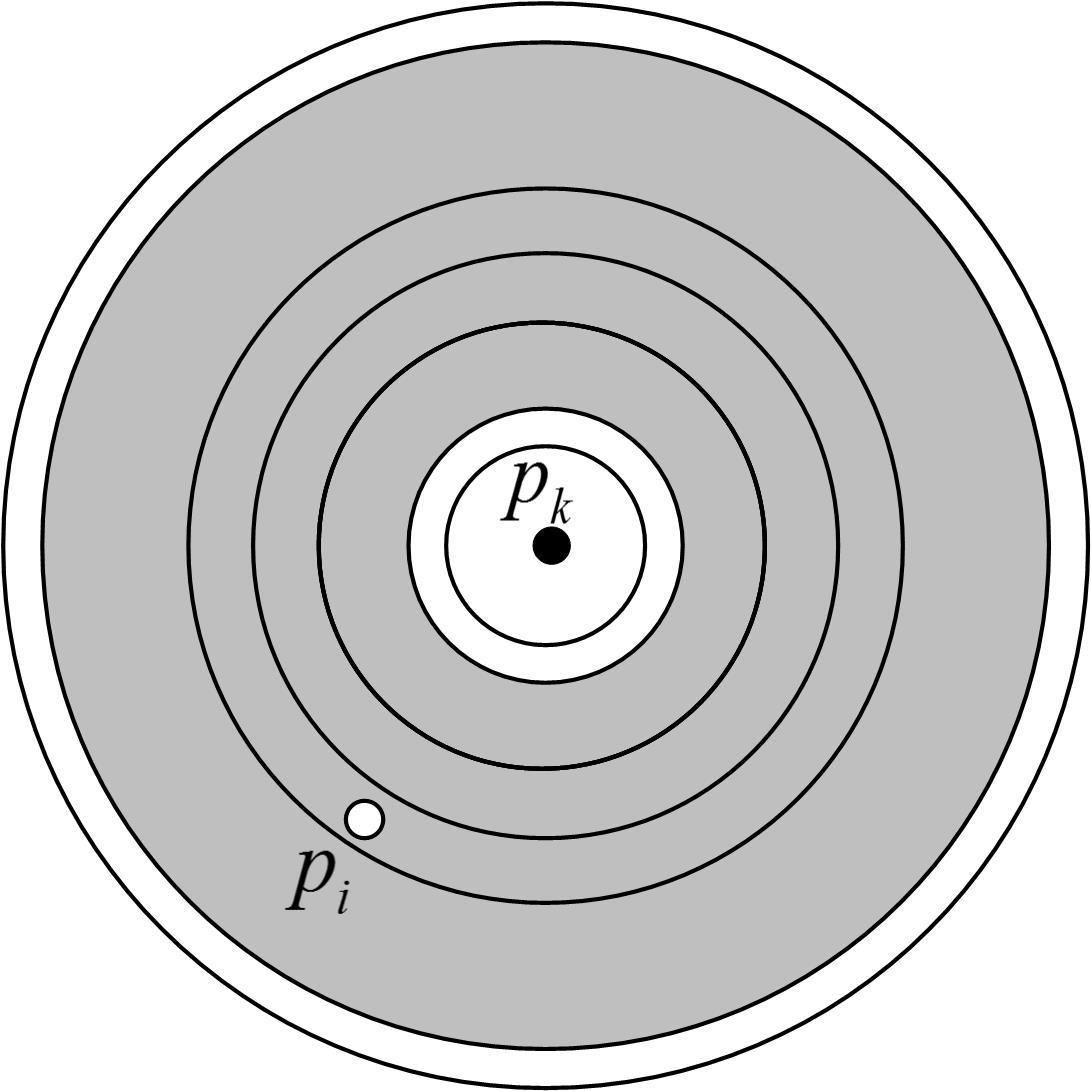}
      ~~~
      \includegraphics[width=0.25\columnwidth]{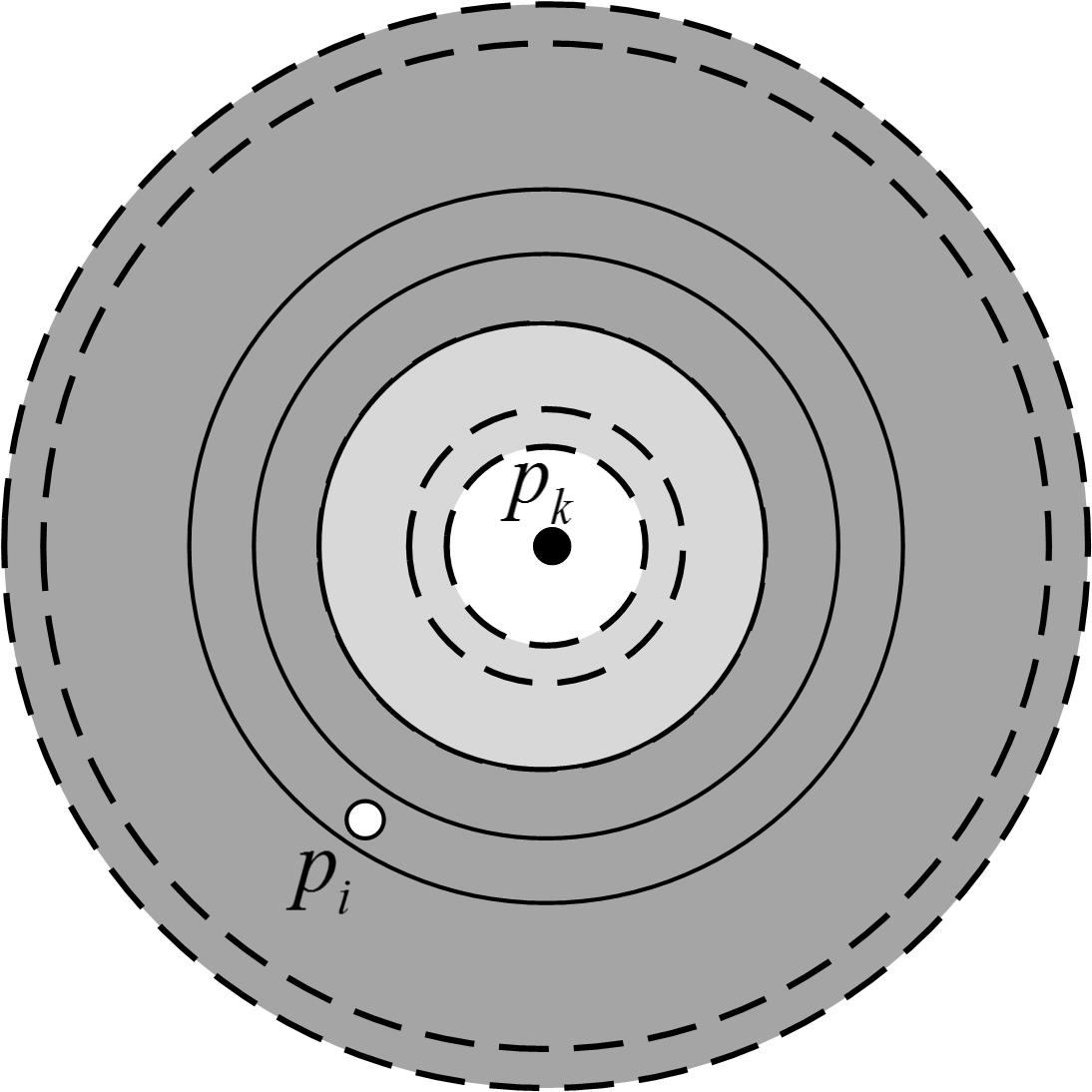}
      \caption{Point $p_i$ with $\centersite(p_i)=p_k^\ell$ in $P_{j-1}$, where $k<j<i$.
      The balls centered at $p_k$ from smaller to larger have radii $c_c\tau^h$, $\frac{c_c\tau}{\tau-1}\tau^h$, $c_r\tau^{m-2}$, $c_r\tau^{m-1}$, $c_r\tau^m$, $\frac{1}{2}c_r\tau^{\ell-1}$, and $c_r\tau^{\ell-1}$, where $m\coloneqq\lceil\log_\tau(\dist(p_i,p_k)/c_r)\rceil$.
      Left: a jump from $p_k^\ell$ to $p_k^h$ in a local net-tree on $P_{j-1}$.
      Center: the shaded annulus does not contain any points of $P_{j-1}$ and corresponds to the previous jump.
      Right: if $p_j$ causes a split below touch on $p_i$, then it is in the light shaded region.
      If $p_j$ causes a split above touch on $p_i$, then it is in the dark shaded region.
      The annuli defined by the two consecutive dashed balls are not necessarily empty in $P_{j-1}$.}
      \label{fig:fig4}
    \end{figure}

    In the following, we divide split touches into two categories of above and below.
    Also, we define two types of split events for the points in a permutation.
    Since these events only depend on the ordering of the points, and not the tree structure, we apply a backwards analysis to find the expected number of such events.
    Then, we show that these events can be used to bound the number of split touches.

    Let $p_k^\ell\coloneqq\centersite(p_i)$ in $P_{j-1}$, where $k<j<i$, and $p_k^\ell$ be the top of a jump with the bottom node $p_k^h$, where $\ell\ge h+2$ (by the definition of a jump).
    Also, let $m\coloneqq\lceil\log_\tau(\dist(p_i,p_k)/c_r)\rceil$.
    By the definition of a center, $h+1\le m\le\ell$.
    Then, the insertion of $p_j$ at some level $f$, where $h\le f\le\ell-1$, results that jump to be split at lower levels.
    After splitting a jump, either $p_i$ stays in the same cell or it moves to the new cell of the new created node for $p_k$ in a lower level.
    When $p_i$ changes its center to the new node of $p_k$, we call that touch a \emph{split above}.
    A split above touch implies that $p_j$ will no longer touch $p_i$ via a split touch.
    If $p_i$ remains in the cell of $p_k^\ell$, then we call that touch a \emph{split below}, see Fig.~\ref{fig:splittouch}.

    Now, we define two types of split events.
    A split below event $\splitbelowevent_{i,j}$ is defined as follows.
    There is a bunch $\bunch$ near $p_i$, for some constant values of $\alpha$ and $\beta$, and $p_j$ is the unique farthest point in $\bunch$ to the first point in that bunch.
    Also we define a split above event $\splitaboveevent_{ij}$ as follows.
    There is a bunch $\bunch$ near $p_i$, for some constant values of $\alpha$ and $\beta$, and $p_j$ is the unique closest point not in the bunch to the first point in $\bunch$.
    In the following lemmas, we find the expected number of these events.

    \begin{lemma}
    \label{lem:exp_split_below_event}
      The expected number of split below events for a point $p_i$ in a permutation $\langle p_1,\ldots,p_n\rangle$ is $O(\rho^{O(1)}\log n)$.
    \end{lemma}
    \begin{proof}
      Let $F(q,\bunch)$ be a random event that $q\in \bunch$ proceeds all points of $\bunch\setminus\{q\}$.
      Also let $\bunch_1,\ldots,\bunch_d$ be the bunches near $p_i$ containing more than one point.
      By Lemma~\ref{lem:constant_bunches}, $d$ is a constant.
      So,
      \begin{align*}
        Pr\big[\splitbelowevent_{i,j}\big]
        &=\sum_{c=1}^d Pr\big[\splitbelowevent_{i,j}\mid p_j\in \bunch_c\big]Pr\big[p_j\in \bunch_c\big]\\
        &=\sum_{c=1}^d\bigg(\sum_{q\in \bunch_c} Pr\big[\splitbelowevent_{i,j}\mid F(q,\bunch_c)\big]Pr\big[F(q,\bunch_c)\big]\bigg)Pr\big[p_j\in \bunch_c\big] \\
        &\le \sum_{c=1}^d\bigg(\sum_{q\in \bunch_c}(\frac{1}{|\bunch_c|-1})(\frac{1}{|\bunch_c|})\bigg)\bigg(\frac{|\bunch_c|}{j}\bigg)
        =\sum_{c=1}^d (\frac{1}{|\bunch_c|-1})(\frac{|\bunch_c|}{j})\le \frac{2d}{j}.
      \end{align*}
      Therefore, $\sum_{j=1}^{i-1}Pr\big[\splitbelowevent_{i,j}\big]\le 2d\sum_{j=1}^n 1/j=O(\rho^{O(1)}\log n)$.
    \end{proof}

    \begin{lemma}
    \label{lem:exp_split_above_event}
      The expected number of split above events for a point $p_i$ in a permutation $\langle p_1,\ldots,p_n\rangle$ is $O(\rho^{O(1)}\log n)$.
    \end{lemma}
    \begin{proof}
      We define $F(q,\bunch)$ similar to the proof of Lemma~\ref{lem:exp_split_below_event}.
      Let $\bunch_1,\ldots,\bunch_d$ be the bunches near $p_i$.
      Then,
      \begin{align*}
        Pr\big[\splitaboveevent_{i,j}\big]
        =\sum_{c=1}^d \sum_{q\in \bunch_c} Pr\big[\splitaboveevent_{i,j}\mid F(q,\bunch_c)\big]Pr\big[F(q,\bunch_c)\big]
        \le \sum_{c=1}^d\sum_{q\in \bunch_c}(\frac{1}{j-1})(\frac{1}{|\bunch_c|})\le \frac{d}{j-1},
      \end{align*}
      which implies $\sum_{j=2}^{i-1}Pr\big[\splitaboveevent_{i,j}\big]\le d\sum_{j=1}^n 1/j=O(\rho^{O(1)}\log n)$.
    \end{proof}

    In the following, we show that how split touches can be counted by the previous events.

    \begin{theorem}
    \label{thm:exp_split_below_touch}
      The expected number of split below touches in a random permutation is $O(\rho^{O(1)}n\log n)$.
    \end{theorem}
    \begin{proof}
      In order to relate split below touches to split below events, we should specify when such touches occur in a tree and then find the right bunches for the corresponding events.
      First, we show that if $p_j$ touches $p_i$ with a split below touch then either $p_j$ belongs to a bunch near $p_i$ or it is followed by a basic touch.
      We have $h+2\le m\le\ell$, because for $m=h+1$ we always have split above touches (the cell of $p_i$ will be changed).
      Now, we have two cases: either $m=\ell$ or $h+2\le m\le\ell-1$.
      \begin{enumerate}
        \item[(a)]
        If $m=\ell$, then the insertion of $p_j$ at any level between $h$ and $\ell-1$ results a split below touch, i.e. $h\le f\le \ell-1$.
        If $f=\ell-1$, then $p_j$ will also touch $p_i$ with a basic touch, so we can charge the basic touch to pay the cost of the split below touch.
        If $f=\ell-2$, then either $p_j$ requires to be promoted to level $\ell-1$ or it stays at the same level.
        If $p_j$ is promoted to level $\ell-1$, then it touches $p_i$ with a basic touch, so we can charge the basic touch to pay the cost of the split below touch.
        If $p_j$ stays at level $\ell-2$, then no further promotion is needed for $p_j$, which implies that the highest level of $p_j$ after the insertion of all $n$ points will be $\ell-2$.
        Therefore, each point may fall in this situation at most once, as such the total number of split below touches for all $n$ points falling in this category is $O(n)$.

        Now, let $h\le f \le \ell-3$.
        In this case, $h\le\ell-3=m-3$.
        From Lemma~\ref{lem:empty_annulus} and $c_r\tau^{m-1}<\dist(p_i,p_k)\le c_r\tau^m$, \[P_{p_k^h}\subset\ball(p_k,\frac{1}{2}c_r\tau^h)\subset\ball(p_k,\frac{1}{2}c_r\tau^{m-3})\subset\ball(p_k,\frac{1}{2\tau^2}\dist(p_i,p_k))\]
        and
        \[[\ball(p_k,\frac{1}{2\tau}\dist(p_i,p_k))\setminus\ball(p_k,\frac{1}{2\tau^2}\dist(p_i,p_k))]\cap P_{j-1}=\emptyset.\]
        Also,
        \[\dist(p_j,p_k)\le c_r\tau^f\le c_r\tau^{m-3}<\frac{1}{\tau^2}\dist(p_i,p_k).\]
        Since for $\tau>2$, we have $\frac{1}{2\tau^2}<\frac{1}{\tau^2}<\frac{1}{2\tau}$, if we set $\alpha=\frac{1}{\tau^2}$ and $\beta=\frac{1}{2\tau}$ in Definition~\ref{def:bunch}, then $p_j$ will be in a bunch near $p_i$.
        \item[(b)]
        If $h+2\le m\le\ell-1$, then the insertion of $p_j$ at any level between $h$ and $m-2$ results a split below touch, i.e. $h\le f \le m-2$, see Fig.~\ref{fig:fig4}.
        Using Lemma~\ref{lem:empty_annulus} and $c_r\tau^{m-1}<\dist(p_i,p_k)\le c_r\tau^m$, \[P_{p_k^h}\subset\ball(p_k,\frac{1}{2}c_r\tau^h)\subset\ball(p_k,\frac{1}{2}c_r\tau^{m-2})\subset\ball(p_k,\frac{1}{2\tau}\dist(p_i,p_k))\]
        and
        \[[\ball(p_k,\frac{1}{2}\dist(p_i,p_k))\setminus\ball(p_k,\frac{1}{2\tau}\dist(p_i,p_k))]\cap P_{j-1}=\emptyset.\]
        Also,
        \[\dist(p_j,p_k)\le c_r\tau^f\le c_r\tau^{m-2}<\frac{1}{\tau}\dist(p_i,p_k).\]
        For $\tau>2$, we have $\frac{1}{2\tau}<\frac{1}{\tau}<\frac{1}{2}$.
        Therefore, if we set $\alpha=\frac{1}{\tau}$ and $\beta=\frac{1}{2}$ in Definition~\ref{def:bunch}, then $p_j$ belongs to a bunch near $p_i$.
      \end{enumerate}
      So far, we proved that if $p_j$ causes a split below touch on $p_i$ and is not followed by a basic touch, then $p_j$ is in a bunch near $p_i$.
      However, $p_j$ is not necessarily the farthest point in the bunch.
      If $\dist(p_j,p_k)>\frac{c_c}{\tau-1}\tau^{h+1}$, then by Lemma~\ref{lem:covering}, $p_j$ is the unique farthest point to $p_k$.
      If $\dist(p_j,p_k)\le c_c\tau^h$, then $p_j$ is in the subtree rooted at $p_k^h$ and should not be promoted to a level greater than $h-1$, so $p_j$ will never results a split touch on $p_i$.
      The remaining case is when $c_c\tau^h<\dist(p_j,p_k)\le c_c/(\tau-1)\tau^{h+1}$, see Fig.~\ref{fig:fig4}.

      For all points in $P_j$ of a distance in $(c_c\tau^h,\frac{c_c}{\tau-1}\tau^{h+1}]$ from $p_k$, only one can touch $p_i$ via a split below touch, because the first point creates a node $p_k^{h+1}$ and touches $p_i$ and the remaining will be added as children of $p_k^{h+1}$, so they will not touch $p_i$ with a split below touch.
      In this case, we charge the farthest point to pay the cost of this split below touch.
      This argument is valid only if $p_k^{h+1}$ is not removed later, and this removal happens when the only child of $p_k^{h+1}$ finds a closer parent.
      By the triangle inequality, we can easily show that $p_k^{h+1}$ and the new parent are relatives, so $p_k^{h+1}$ should remain in the tree.
      Thus, either the farthest point touches $p_i$ via a split below touch or it pays for another point that touches $p_i$.

      Also note that the order of points in a permutation specifies the first point of a bunch, and it is important because one of its associated nodes is the center of $p_i$ in $P_{j-1}$.
      Furthermore, using Lemma~\ref{lem:const_touch}, the maximum number of split below touches from $p_j$ on $p_i$ is constant.
      In conclusion, the expected number of split below touches on $p_i$ is bounded by the summation of the expected number of split below events (Lemma~\ref{lem:exp_split_below_event}) and basic touches (Theorem~\ref{thm:expected_basic_touches}).
    \end{proof}

    \begin{theorem}
    \label{thm:exp_split_above_touch}
      The expected number of split above touches in a random permutation is $O(\rho^{O(1)}n\log n)$.
    \end{theorem}
    \begin{proof}
      Our aim is to make a connection between split above touches and events.
      For a split above touch we have $h+1\le m\le \ell-1$ and $m-1\le f\le \ell-1$, see Fig.~\ref{fig:fig4}.
      If $h>-\infty$, then the subtree rooted at $p_k^h$ has more than one point.
      Therefore, before $p_j$ touches $p_i$ via a split above touch, $p_i$ should have been touched with a split below touch by a node in level $h-1$.
      We charge that split below to pay the cost of split above touches on $p_i$ for levels $h$, $h+1$, and $h+2$.
      In other words, if $h\le f\le h+2$, then the cost of split above touches has been already paid by an earlier split below touch.
      As such, we only need to handle the remaining split above touches when $f\ge h+3$ and $h+4\le m\le\ell-1$.

      Using Lemma~\ref{lem:empty_annulus} and $c_r\tau^{m-1}<\dist(p_i,p_k)\le c_r\tau^m$, \[P_{p_k^h}\subset\ball(p_k,\frac{1}{2}c_r\tau^h)\subset\ball(p_k,\frac{1}{2}c_r\tau^{m-4})\subset\ball(p_k,\frac{1}{2\tau^3}\dist(p_i,p_k))\]
      and
      \[[\ball(p_k,\frac{1}{2}\dist(p_i,p_k))\setminus\ball(p_k,\frac{1}{2\tau^3}\dist(p_i,p_k))]\cap P_{j-1}=\emptyset.\]
      If $p_j$ is promoted to level $f$, then $p_j^{f-1}$ violates the covering property and as such, \[\dist(p_j,p_k)>c_c\tau^f\ge c_c\tau^{m-1}=\frac{\tau-4}{2\tau}c_r\tau^{m-1}\ge\frac{1}{10\tau} c_r\tau^{m}\ge\frac{1}{10\tau}\dist(p_i,p_k).\]
      Otherwise, if $p_j$ is directly inserted into level $f$, $p_j$ cannot be a relative of $p_k$ at level $f-1$, as such $\dist(p_j,p_k)>c_r\tau^{f-1}\ge c_r\tau^{m-2}\ge\frac{1}{\tau^2}\dist(p_i,p_k)$.
      Since $\tau\ge 5$, $\dist(p_j,p_k)>\min\{\frac{1}{10\tau},\frac{1}{\tau^2}\}\dist(p_i,p_k)\ge\frac{1}{2\tau^2}\dist(p_i,p_k)$.
      Also, we have $\frac{1}{2\tau^3}<\frac{1}{2\tau^2}<\frac{1}{2}$, so if we set $\alpha=\frac{1}{2\tau^3}$ and $\beta=\frac{1}{2\tau^2}$ in Definition~\ref{def:bunch}, then $p_j$ does not belong to any bunch near $p_i$.

      From Lemma~\ref{lem:empty_annulus}, the distance of every point of $P_{j-1}$ not in $P_{p_k^h}$ to $p_k$ is greater than $\frac{1}{2}c_r\tau^{\ell-1}$.
      We know that if $p_j$ touches $p_i$ via a split above touch, then $\dist(p_j,p_k)\le c_r\tau^{\ell-1}$.
      So, if $\dist(p_j,p_k)\le \frac{1}{2}c_r\tau^{\ell-1}$, then $p_j$ is the closest point to $p_k$ not in that bunch, as desired.
      Otherwise, if $\frac{1}{2}c_r\tau^{\ell-1}<\dist(p_j,p_k)\le c_r\tau^{\ell-1}$, we use a charging argument similar to the proof of Theorem~\ref{thm:exp_split_below_touch}.
      In other words, if there are many points of a distance in $(\frac{1}{2}c_r\tau^{\ell-1},c_r\tau^{\ell-1}]$ from $p_k$, then only one of them results a split above touch, and it is the first one that splits the jump at level $\ell-1$.
      Therefore, we can charge the closest point to $p_k$ not in the bunch to pay the cost of that split above touch.
      Notice that because $p_k^{\ell-1}$ is relative to the point that creates it, $p_k^{\ell-1}$ will not be removed later and will be in the output.

      For $h=-\infty$, i.e. the subtree rooted at $p_k^h$ only contains $p_k^h$, we can use a similar argument (without charging split below touches).
      In a nutshell, the expected number of split above touches on $p_i$ is bounded by the expected number of split above events (Lemma~\ref{lem:exp_split_above_event}) and split below touches (Theorem~\ref{thm:exp_split_below_touch}) on $p_i$.
    \end{proof}

    \begin{theorem}
    \label{thm:exp_merge_touch}
      The expected number of merge touches in a random permutation is $O(\rho^{O(1)}n\log n)$.
    \end{theorem}
    \begin{proof}
      Recall that the insertion of $p_j$ results a merge touch on $p_i$ if $p_k^\ell\coloneqq\centersite(p_i)$ is deleted from the tree and $p_i$ moves to the cell of $\parent(p_k^\ell)$.
      Let $p_k^\ell$ be added to the tree after the insertion of a point $p_a$, where $a<j$.
      Since $p_i$ is in the cell of $p_k^\ell$, $p_i$ belongs to the cell of $\parent(p_k^\ell)$ before $p_a$ is inserted.
      Therefore, the insertion of $p_a$ results a split above touch on $p_i$.
      In other words, before $p_j$ touches $p_i$ via a merge touch, there exists another point $p_a$,where $a<j$ , such that it touches $p_i$ via a split above touch.
      Therefore, it suffices to pay the cost of a merge touch when an earlier split above touch occurs.
      Thus, the total number of merge touches is bounded by the total number of split above touches.
    \end{proof}




  \section{Conclusion} 
\label{sec:conclusion}
  In this paper, we proposed local net-trees as a variation of net-trees with much easier to maintain properties.
  We proved that local net-trees are also net-trees with a slightly different parameters.
  Then, we presented a simple algorithm to construct local net-trees incrementally from an arbitrary permutation of points in a doubling metric space.
  We relegated the challenge of achieving $O(n\log n)$ time complexity to the analysis part.
  To analyze our algorithm, we defined a notion of touches corresponding to the number of distance computations, and proved that the total expected number of touches in a permutation is $O(n\log n)$.


  \clearpage
  \bibliographystyle{abbrv}
  \bibliography{bibliography}
\end{document}